\documentclass[12pt, oneside]{amsbook}
\usepackage{amsmath}
\usepackage{amssymb}
\usepackage{amscd}
\usepackage[dvips]{graphicx}
\usepackage{xy}
\input xy
\xyoption{all}
\usepackage[letterpaper,text={5.75in,8.20in},left=1.5in,top=1.3in,ignorehead]{geometry}

\theoremstyle{plain}

\newtheorem{thm}{Theorem}[chapter]

\newtheorem{example}[thm]{Example}
\newtheorem{prop}[thm]{Proposition}
\newtheorem{conjecture}[thm]{Conjecture}
\newtheorem{problem}[thm]{Problem}
\newtheorem{question}[thm]{Question}

\newcommand{\N}{\mathbb{N}}

\newcommand{\Z}{\mathbb{Z}}

\renewcommand{\hat}[1]{\widehat{#1}}

\begin{document}

\frontmatter


\begin{titlepage}
\begin{center}
{ \LARGE \textbf{ Traversals of Infinite Graphs with Random Local Orientations}}

\vspace*{20mm} {\large { By\\
                       David White}}

\vspace*{25mm} {\large { Adviser: Danny Krizanc, \\
                         Professor of Computer Science}}\\

\vspace*{20mm} {\normalsize { Wesleyan University \\
                         Middletown, CT \\
                         May 2012}}\\

\vspace*{20mm}{\normalsize \em{A Thesis in Computer Science\\
              submitted in partial fulfillment of the\\
              requirements for the degree of Master of Arts}}
\end{center}
\end{titlepage}





\chapter*{Abstract}
We introduce the notion of a random basic walk on an infinite graph, give numerous examples, list potential applications, and provide detailed comparisons between the random basic walk and existing generalizations of simple random walks. We define analogues in the setting of random basic walks of the notions of recurrence and transience in the theory of simple random walks, and we study the question of which graphs have a cycling random basic walk and which a transient random basic walk.

We prove that cycles of arbitrary length are possible in any regular graph, but that they are unlikely. We give upper bounds on the expected number of vertices a random basic walk will visit on the infinite graphs studied and on their finite analogues of sufficiently large size. We then study random basic walks on complete graphs, and prove that the class of complete graphs has random basic walks asymptotically visit a constant fraction of the nodes. We end with numerous conjectures and problems for future study, as well as ideas for how to approach these problems.


\chapter*{Acknowledgements}

First, I give my sincere thanks to my advisor, Danny Krizanc, for his advice, encouragement, fascinating questions and conversations, for helping me find this thesis topic, and for his guidance as I wrote this document. It has been a great joy to work with him. Thanks also go to the other members of my committee, Eric Aaron and Mike Keane, for their feedback on my thesis.

I would like to thank Michel Dekking and Mike Keane for helping me understand this problem, teaching me background material on random walks, and for asking questions which guided me to the results in Chapter \ref{sec:locallyfinite}. Similarly, I thank Sunil Shende for his help with making the experiments which led to the statements of the theorem and conjecture in Chapter \ref{sec:complete}, and I thank Leszek Gasieniec for an advance preprint of \cite{supersize}.

There are other members of the Wesleyan community who I also must thank, for without them I would have never been able to complete this thesis. To my fellow grad students for all their support over the years and for making Middletown a better place to live. To Mike Rice for showing me that being a mathematician and doing computer science research need not be mutually exclusive, and for helping me figure out how to balance these two passions. To Mark Hovey for his mentoring, his support, and for teaching me how to do research at the PhD level. I feel truly blessed to have been mentored by two such incredible individuals as Professors Hovey and Krizanc.

Finally, I need to thank my family for their unending support over the years. To my mother, Vicki Quade, for always being willing to listen and give advice. To my father, Charlie White, for his wisdom and calm support as I find my way through graduate school. Without my family, I would not be the person I am today, and I thank them most of all for that.


\tableofcontents 

\mainmatter \pagenumbering{arabic}

\chapter{Introduction}
\label{introduction}

Many places in computer science it is valuable to represent a system as a graph and to represent change in the system as a random process on that graph. For instance, the nodes of the graph might be positions and the edges might be allowable directions to move, or the nodes might be computers in a network and the edges might be routes by which information can travel.

The most well-studied random process on a graph is the simple random walk, where directions to move or routes for information travel are chosen uniformly at random from the set of possible options. When the graph in question is infinite, an important question is whether the simple random walk is recurrent or whether it is transient (these terms are defined in Section 1.\ref{sec:randomwalks}). Numerous generalizations of simple random walks have been studied over the years, and this question is always a driving force behind the research into these generalizations.

In recent years researchers have begun to study a random process called the random basic walk (discussed in Section 1.\ref{sec:basic}), which is related to the simple random walk. The random basic walk may be applied to both of the examples mentioned above. Existing literature only studies the random basic walk on finite graphs, with a particular focus on grids of the form $G_{k,n}$, with $k$ rows and $n$ columns.

We will study the random basic walk on infinite grids and then on much more general infinite graphs. We first define the properties analogous to recurrence and transience (called cycling and transience respectively) and give several examples to demonstrate the types of interesting behavior which can occur in a random basic walk. We then study the question of which graphs have a cyclic random basic walk and which have a transient random basic walk, focusing in the process on the dichotomy between graphs of bounded degree and graphs of unbounded degree. En route we prove that cycles of arbitrary length are possible for regular graphs, but are unlikely. We end by returning to the question of the random basic walk on finite graphs, including new results on complete graphs $K_n$ which were suggested by our proof methods in the infinite situation.

\section{Background on Simple Random Walks}
\label{sec:randomwalks}

Simple random walks have been studied since at least 1905 (\cite{pearson}), and there are numerous sources which treat them (e.g. \cite{snell, Spit}). We include the basic definitions here for completeness. Let $G$ be a connected graph with countably many vertices and edges, and assume $G$ is locally finite, i.e. every vertex has only finitely many neighbors. We now describe a process which creates a path in $G$. Select a starting vertex $v_0$. Select a neighbor $v_1$ of $v_0$ uniformly at random, i.e. the probability a neighbor will be selected is $1/deg(v_0)$. Next, select a neighbor $v_2$ of $v_1$ uniformly at random, and continue in this way. This process results in a path $v=(v_0,v_1,v_2,\dots)$, where for every $i$, there is an edge $v_i \sim v_{i+1}$. Let $\Omega = \{(v_0,v_1,\dots)\;|\; \forall i \; v_i \sim v_{i+1}\}$ be the set of such paths with starting vertex $v_0$. It is clear from our construction how to compute the probability of any given finite path $\gamma$:

$$P(\gamma = (v_0,v_1,\dots,v_n)) = \left(\frac{1}{deg(v_0)}\right)\left(\frac{1}{deg(v_1)}\right) \cdots \left(\frac{1}{deg(v_{n-1})}\right)$$

There is a $\sigma$-algebra $\mathfrak{a}$ on $\Omega$ generated by these finite paths (i.e. the cylinder sets). Equivalently, $\mathfrak{a}$ is given by the Borel subsets of the compact space $\Omega$. This $\sigma$-algebra makes $P$ into a probability measure on $\Omega$. A \emph{simple random walk} is a tuple $(G,\Omega,\mathfrak{a}, P)$ constructed as above. We will sometimes abuse notation and refer to an individual path $v\in \Omega$ as a simple random walk, since this is the walk an individual particle takes.

A simple random walk $(v_0,v_1,\dots)$ is \emph{recurrent} if $P(\exists \; n\geq 1\;|\; v_n=v_0)=1$, i.e. if it returns to the starting vertex at some point. A simple random walk is \emph{transient} if it is not recurrent. The term transient is used because at every step the walk has a non-zero probability of escaping to infinity (i.e. of never returning). The following statements are well-known:

\begin{prop}
\begin{enumerate}
\item Recurrence is equivalent to the statement that with probability 1 the simple random walk returns to $v_0$ infinitely many times, since a simple random walk is Markovian.
\item Recurrence is equivalent to the statement that with probability 1 every vertex is visited infinitely often, by a Borel-Cantelli argument.
\item Transience is equivalent to the statement that with probability 1 every vertex is visited only finitely often.
\item A simple random walk on any connected graph is either recurrent or transient.
\end{enumerate}
\end{prop}

In 1921, Georg P\'{o}lya \cite{Polya} famously studied simple random walks on lattices $\Z^d$ and characterized when such walks are recurrent and when they are transient:

\begin{thm}[P\'{o}lya's Theorem]
\label{polya}
The simple random walk on $\Z^d$ is recurrent for $d\leq 2$, but transient for all $d>2$.
\end{thm}

P\'{o}lya's original proof proceeded by writing out and bounding the probabilities of certain paths in $\Omega$. Modern proofs have removed much of this computational aspect. For example, there is an elegant proof using Martingales in \cite{Doob}, and there is a very clever proof using the theory of electrical networks in \cite{snell}. We will prove an analogue of this theorem for the random basic walk in Section \ref{ch:infinite}.\ref{sec:lattices}, and the proof will be more along the lines of P\'{o}lya's original proof, except that the bounds will be much easier to compute.

Simple random walks are also studied on finite graphs (see \cite{lovasz}), and the main questions there regard cover time, hitting time, mixing time, and load balancing. The cover time is the expected amount of time it takes a simple random walk to visit all vertices. The partial cover time is the expected amount of time it takes to visit a constant fraction of the nodes. The edge cover time measures the expected number of steps until all edges have been traversed. The analogues for these three concepts on random basic walks will be studied in Chapter \ref{sec:complete}. The hitting time is the expected time of first return to the starting vertex. This concept does not appear to have a natural generalization to random basic walks, since they need not ever return to the starting vertex, even on a finite graph. The mixing time measures how quickly the simple random walk converges to its limit distribution. This concept has not been studied at all for the random basic walk. Load balancing refers to how many times an individual node is visited relative to the number of times other nodes are visited. This concept has not been studied for the random basic walk.

The finite analogue of $\Z^2$ is the graph $[n]\times [n]$ where $[n] = \{1,2,\dots,n\}$. We will discuss this graph more in the next section, but for now we simply record that the hitting time of a simple random walk on $[n]\times[n]$ is $O(|V|\log^2(|V|))$ according to Section 11.3.2 of \cite{mixing}. For general graphs, the cover time has been bounded by $O(|V||E|)$ in \cite{covering2}. For regular graphs this has been improved to $O(|V|^2)$ in \cite{covering10}. For regular expander graphs this has been improved to $\Theta(|V|\log|V|)$ in \cite{covering3}. For an arbitrary graph, a general lower bound of $(1-o(1))|V|\ln|V|$ has been obtained in \cite{covering11}. The edge cover time has been proven to be at least $\Omega(|E|\log|E|)$ and at most $O(|V||E|)$ in \cite{covering9, covering12}. See also: \cite{covering1, covering4, covering5, covering6, covering7, covering8}.

\section{Background on Basic Walks}
\label{sec:basic}

A common problem in computer science is that of graph exploration by a mobile entity. For example, software moving on a network of computers or a webcrawler moving on the internet graph. An example which we will carry with us throughout this document is that of a robot exploring an unknown terrain. This problem has been studied for many years, and several deterministic answers have been proposed. However, many of these solutions depend on sophisticated sensors to guide the robot or depend on sophisticated algorithms to direct the robot's movement. In practice, these solutions can be infeasible if one cannot afford the expensive sensors (e.g. GPS, infrared sensors, ultrasound sensors), if time and space constraints inside the robot make the algorithms infeasible, or if the area to be explored is quite large. For this reason, attention has shifted recently to solutions making use of randomness. Random solutions often give suboptimal performance, but with significant savings on time and space requirements.

In recent years a robot vacuum cleaner called the iRobot Roomba has become popular, but these robots sometimes get trapped in corners or behind furniture. Makers of the Roomba wish it to explore the entire room without getting trapped and without covering the same ground too many times. A discrete approximation to the continuous setting of the room is a graph, where the vertices are allowable positions of the robot and edges are allowable directions to move. Since most rooms are rectangular-shaped, it is natural to study this process on a two dimensional grid $G_{k,n}$ with $k$ rows and $n$ columns, as is done for instance in \cite{Sir07}. Letting $k$ and $n$ go to infinity yields the graph $\Z^2$ and lattices $\Z^d$ more generally. Research on these graphs answers asymptotic questions on $G_{k,n}$ and also a better approximation to the continuous situation via taking a finer and finer mesh on $[n]\times [n]$. The graphs $G_{k,n}$ and $\Z^2$ represent empty rooms. To study graph exploration in the presence of furniture and other obstacles one must study more general graphs.

The question remains of how the robot chooses which edge to follow from a given vertex. If the robot make a choice of direction uniformly at random from the set of possible directions, then the path of the robot will be a simple random walk. However, this decision scheme does not make use of the existence of memory in robots, so it is natural to wonder if one can do better. Many different decision models have been proposed, e.g. in \cite{robot1, robot2, robot3, robot4, robot5, robot7, robot14, robot8, robot9, robot10, robot11, robot12, robot13}. We will focus on the model considered in \cite{Sir09, Sir07, robot15}, which helps the robot by placing pointers on all edges and which makes use of only a constant number of bits of robot memory. In this model $G$ is a directed graph such that whenever there is an arc $v\to w$ there is also an arc $w\to v$.

We now describe this model. At every vertex $v$ order and label the outgoing arcs $\{e_1,e_2,\dots,e_{deg(v)}\}$ by consecutive integers $1,2,\dots,deg(v)$. When the robot enters $v$ by an arc labeled $i$ it will exit by the arc labeled $(i \mod{deg(v)})+1$. Denote the label of $e_i$ by $L(e_i)$. The collection of labels is sometimes referred to as the \emph{port numbers} at $v$, and the configuration of outgoing arcs and labels is the \emph{port orientation} at $v$. The collection of port orientations at all vertices is called a \emph{labeling }on the graph $G$. It is worth noting that the labels on arcs $v\to w$ and $w\to v$ need not match. Once a labeling has been assigned a robot is placed on the graph, i.e. given a starting location $v$ and a starting port $i$. The path of the robot is called a \emph{basic walk}. The basic walk can be thought of as an automaton, with states $(v,i)$ where $v$ is the current position and $i$ is the next port to use. The transition function takes $(v,i)$ to $(w,j)$ where $w$ is the vertex which $i$ points to and $j$ is $(i\mod deg(w))+1$.

One immediate observation about the basic walk is that when vertices are visited multiple times it is possible for the robot to get trapped in a cycle (see Example \ref{latticetraps}). In \cite{Sir09} and \cite{Sir07} the authors address the question of how to find a good labeling on a finite graph so that all nodes are visited in a periodic manner. They model the robot as a finite state automaton and therefore \cite{budach} shows that there are labelings which cause the robot to fail to visit all vertices. The main result of \cite{Sir09} is an algorithm to set the port orientations at all vertices to create a period of $O(3.5*|V|)$ with which the robot explores the graph, an improvement over the simple random walk's cover time. However, this paper requires the port numbers to be updated while the walk is occurring and thus forces the individual nodes in the graph to do computation, provide memory storage, and give the robot much more information than simply the local port orientation.

One way to reduce the amount of computation above is to set the port numbers uniformly at random before the walk starts, i.e. select label $L(e_1)$ uniformly at random from $\{1,2,\dots,deg(v)\}$, then select $L(e_2)$ from $\{1,2,\dots,deg(v)\}-\{L(e_1)\}$, and continue in this way until all arcs are labeled. With this labeling the path of the robot is called a \emph{random basic walk}, and was first considered in \cite{supersize}. It does not appear that this process has been studied when a distribution other than the uniform distribution is used to select the labels. Note that once the labels, starting vertex, and starting port are fixed, the random basic walk is completely determined because the robot's transition function is completely deterministic. The randomness only comes into play with the choice of labels on the arcs. Thus, probabilistic statements are made with respect to the labeling process, and we often consider certain port orientations which occur with some probability.

In \cite{supersize} the authors consider the analogue of cover time for the random basic walk. Due to the existence of traps, the correct question to ask in this context is how much of the graph one can expect the robot to explore before becoming trapped, rather than how long it will take to explore the entire graph. The authors prove that for the class of graphs $G_{k,n}$ with $k$ fixed and $n\to \infty$, the expected length of the longest tour is $\Theta(\log n)$. They then give experimental evidence that on $G_{n,n}$ the expected maximum length of a cycle is $1.2701 \cdot |V|^{1.8891}$. The existence of such a supersize tour is surprising because the preponderance of small cycles discovered in Section \ref{ch:infinite}.\ref{sec:lattices} suggests that a robot would need to be very lucky to explore this much of the graph. There are many questions still open about finite graphs, e.g. the expected average length of a cycle, the expected maximum length of a cycle, and the expected number of vertices visited before a cycle is hit.

Another natural question to ask about the random basic walk regards recurrence-transience behavior on lattices, since $\Z^2$ is the infinite analogue of the graphs $[n]\times[n]$ above. In light of the ability of the robot to be trapped, P\'{o}lya's original notions of recurrence and transience must be tweaked. A walk which hits a cycle may never return to its origin, but this walk still has the same flavor as a recurrent simple random walk. In this case we say that the random basic walk \emph{cycles}. If a random basic walk does not cycle, then we say it is \emph{transient}, i.e. it visits infinitely many vertices but never visits the same vertex infinitely many times. The pigeonhole principle guarantees that a basic walk which visits a vertex $v$ more than $deg(v)$ times must leave by the same arc more than once, forcing a cycle. Our notion of transience matches P\'{o}lya's, i.e. for each vertex $v$, $\Pr(X_n=v$ for infinitely many $n)=0$. In Chapter \ref{ch:infinite} we will prove an analogue of P\'{o}lya's Theorem for a large class of infinite graphs.

Before progressing, we wish to note that for any graph $G$, the way in which labels are placed on arcs may be changed. Rather than labeling every arc simultaneously at the start, we may define just one new label at each step. Suppose the robot enters vertex $v$ by port $i$ (for the starting vertex an initial port is provided, which we will assume is $1$ in our examples). If $v$ has an outgoing arc labeled by $(i \mod{d(v)})+1$ already then the robot will follow this arc and no labeling will be done. If $v$ does not have such an arc, then there must be a non-empty set of arcs leaving $v$ with no label. Choose one of these arcs uniformly at random and assign it the label $(i \mod{d(v)})+1$. It is easy to see that the walks possible in this definition exactly match the walks possible in the old definition. Furthermore, a given walk will have the same probability in each definition. It will be easier to state and prove results thinking with the new definition. One could also define a third way to assign labels, where the port orientation at a given vertex $v$ is assigned uniformly at random when $v$ is first entered, but we will not need this formulation.


We sketch an argument that the formulations above are equivalent. If all port numbers at $v$ are assigned simultaneously, then each outgoing arc has probability $1/d(v)$ of receiving a fixed port number $i$. We claim that if the port numbers are assigned one at a time, only as needed by the robot, then the same probabilities are achieved. For the first outgoing arc, the probability of receiving any given port number $i_1$ is obviously $1/d(v)$. For the second arc to be labeled, the probability is $\frac{d(v)-1}{d(v)}*\frac{1}{d(v)-1}=\frac{1}{d(v)}$ because we must first know that the port number chosen is not $i_1$, and then we have $d(v)-1$ choices for which $i_2$ will be chosen of the $d(v)-1$ possibilities remaining. For the $k$-th arc to be labeled, we must first know that $i_1,i_2,\dots,i_{k-1}$ are not chosen and then there are $d(v)-k$ possibilities remaining. So the probability is $\frac{d(v)-k}{d(v)}*\frac{1}{d(v)-k}=\frac{1}{d(v)}$. This proves the two random processes are the same, so we are free to think of the assignment of port numbers as occurring one assignment per step of the robot. This viewpoint will make our theorems and proofs much easier to understand.




\section{Applications}

Simple random walks have found numerous applications in computer science, and it is natural to wonder if the random basic walk could be applied in similar situations. The most famous application of simple random walks is probably the Markov Chain Monte Carlo algorithm for solving combinatorial optimization problems (\cite{choice11}). Unfortunately, this application has no hope of a random basic walk generalization because the random basic walk is not Markovian. However, there are many other applications of simple random walks including network routing, rumor routing, searching and query processing on distributed networks, load balancing and self-stabilization of such networks as a way to counteract transmission failure, energy savings on large networks, image processing, exploration of unknown terrain, and clustering.

Many of these applications discussed above rely on simple random walks because of their simplicity of implementation, savings on time and memory, and local nature. The random basic walk shares many of these features, and this section discusses several applications where the random basic walk could also be applied. Furthermore, because both the simple random walk and random basic walk rely on randomness, both should give robust applications, i.e. applications which can survive structural changes caused by failures, sleep modes, etc. As many of these applications make use of facts about cover time and load balancing from the theory of simple random walks, the analogous notions for the random basic walk would need to be developed before one could determine the utility of applying random basic walks in a similar way.

Many of these applications above are discussed in \cite{choice} as potential applications of the simple random walk with choice (RWC) which those authors introduce. The RWC makes use of a small amount of local memory to choose its next direction based on which neighbor has been visited the least. The RWC shares some properties with the random basic walk: randomness, use of local memory, local decision rule, and favoritism for visiting new nodes (which the basic walk has before cycling occurs). Thus, one can hope that problems for which RWC leads to good applications will also admit applications of the random basic walk.

Due to the increasing number of monitoring applications which make use of a large network of small, smart sensors, there is great demand for search and distribution algorithms with small overhead. Several such algorithms make use of random walks. One famous algorithm for searching the internet is PageRank (see \cite{choice12}). Another is topic-sensitive PageRank, which computes the stationary probability distribution coming from a simple random walk on websites (see \cite{choice15}). Another application of random walks to sensor networks is \cite{choice2}, which focuses on robust query processing. This paper suggests an application of the random basic walk because only a constant fraction of the sensor network needs to be visited. Similarly, \cite{choice3} considers algorithms for routing on a sensor network. It introduces rumor routing, which is a compromise between flooding queries and flooding event notifications on a network. Rumor routing works by creating paths which lead to each event, so queries move on the network via a simple random walk to find the event path to the correct event.

Peer-to-peer networks are more general than sensor networks or the internet network. Algorithms for searching peer-to-peer networks are proposed in \cite{choice6} and \cite{choice8}, and make use of $k$ simultaneous simple random walks to find the necessary data. In this setting, the query is the random walker, and making use of randomness leads to savings on bandwidth and energy consumption. \cite{choice1} also considers the problem of saving energy during distributed computation. In this paper, the simple random walk is used to control when transmitting nodes are activated. Further work in \cite{choice5} quantifies the effectiveness of simple random walk methods for searching peer-to-peer networks.

Another paper which is concerned with saving energy during distributed computation is \cite{choice7}. Like \cite{choice1}, this paper allows nodes to switch from active to inactive and vice versa at random times, and it studies routing in this context via constrained random walks on dynamic graphs. As with the random basic walk, these algorithms do not require nodes to maintain state information. This paper goes beyond \cite{choice1} in that load balancing is also studied as a way to save on energy. In order to apply the random basic walk to the dynamic system, one would need to develop a notion of constrained random basic walks. The first step in this direction is the equivalent formulation of the random basic walk discussed in Section 1.\ref{sec:basic}.

Another type of network is an ad-hoc network, which relies on wireless links between entities rather than infrastructure such as telephone lines. Due to partial transmission failure, failure of communication links, and noisy transmission, this is a field where algorithms which can stabilize themselves after a failure are highly valued. In \cite{choice4} the authors use simple random walks to create an algorithm for self-stabilizing communication in ad-hoc networks.

When one studies networks, it is often advantageous to arrange the nodes in a certain way so as to make use of the topology in algorithms. \cite{choice9} focuses on the problem of constructing good topologies on distributed networks, with a focus on networks satisfying certain expander properties. It makes use of simple random walks to deal with where to put new nodes to maintain the properties of the topology. This comes down to using the simple random walk as a sampling algorithm to sample potential places the new node could be put.

Simple random walks are also applied beyond the study of networks. For example, \cite{image} and \cite{choice14} apply simple random walks to segmentation in image processing. In particular, both use hitting time computations: the former to determine which labels should be associated to which pixels, the latter to determine how to assign a keyword/classification to an image as a whole. Another example is \cite{choice13}, which uses the connection between simple random walks and electrical networks discussed in \cite{snell} to give a method for exploring a continuous planar domain which contains very few sensors. This paper relates the cover time to the electrical resistance in the domain. Finally, \cite{choice10} applies simple random walks to the problem of graph clustering via defining a cluster to be the set of vertices visited by a simple random walk before some stopping criterion has been met. This paper uses multiple simultaneous simple random walks for the same purpose.

It is the author's hope that the random basic walk will be useful for some applications similar to those discussed above. Already \cite{Sir07} and \cite{Sir09} have shown that (non-random) basic walks can improve over simple random walks. The benefit of random basic walks over non-random basic walks is savings on overhead and on the memory and computations required by individual nodes. It is likely that more theory will need to be developed before random basic walks can find applications like those above--especially theory related to finite graphs.

\section{Outline of Thesis}

Chapter \ref{ch:history} is devoted to studying the existing literature and previous work on this question. In particular, we include a comparison of the random basic walk with other types of random walks and quasirandom processes. In Chapter \ref{sec:examples} we give numerous examples to better understand how the basic walk works. These examples will be cited throughout the rest of the paper. In Chapter \ref{ch:infinite} we prove an analog of P\'{o}lya's theorem for the random basic walk on lattices $\Z^d$, and discuss consequences for the expected maximum number of vertices visited and for the expected number of vertices visited. We then extend this result to arbitrary graphs of bounded degree, and we also provide an example that proves the bounded degree hypothesis is necessary.

An infinite grid is an approximation to a continuous situation, e.g. a robotic vacuum cleaner vacuuming a room. Thus, we feel the results in Chapter \ref{ch:infinite} have the potential to be very useful for applications. However, the random basic walk was first considered on finite graphs, and so in Chapter \ref{sec:complete} we also consider finite graphs. In \cite{supersize} the authors conjectured that $[n]\times[n]$ was an infinite family of graphs on which the random basic walk can be expected to visit a large fraction of the nodes before becoming trapped. This conjecture is still open, but in Chapter \ref{sec:complete} we prove that the complete graphs $K_n$ are such a family of graphs. We end with several new conjectures and directions for future study.

\chapter{Related Random Processes}
\label{ch:history}

In this chapter we seek to develop an intuition and a historical context for the random basic walk by considering related random processes. There are two key properties of the random basic walk: it has the ability to get trapped, and it's defined via local labels at each vertex. Certain generalizations of random walks have had each of these features before, but none seem to have both features at once. The most notable generalization of the simple random walk which allows the walker to become trapped is the self-avoiding random walk. We'll discuss this random process in the Section \ref{ch:history}.\ref{sec:self} and summarize a few of its key properties. Related to the self-avoiding random walk are reinforced random walks, which we'll also discuss. The most notable random process which is based on local orientations and labels on arcs is the rotor router model discussed in Section \ref{ch:history}.\ref{sec:quasirandom}. We'll define this model, list some of its properties, and explain how it is different from the random basic walk.

We will not require knowledge of self-avoiding random walks, reinforced random walks, or rotor routers in future chapters, so readers who are interested only in new results should skip to the next chapter. The only other novel work in this chapter is in Section \ref{ch:history}.\ref{sec:quasirandom}, where we propose a random rotor router which is related to the rotor routers in exactly the same way that the random basic walk is related to the basic walk. We hope that this notion is of interest to those studying rotor routers and we hope that the methods of Chapter \ref{ch:infinite} prove useful to that study.

We will discover in Chapter \ref{ch:infinite} that the random basic walk has different cycling and transience properties than the random processes considered here, i.e. the intuition for infinite graphs is wrong. Nevertheless, we will discuss the recurrence question for all of these random processes to show how it has driven their study and to demonstrate that this is an important question for any random process which generalizes a random walk. Furthermore, the intuition from this chapter on finite graphs might prove useful for solving the open problems of Chapters \ref{sec:complete} and \ref{conclusion}, even though it was not useful for infinite graphs.

\section{Self Avoiding Random Walks}
\label{sec:self}

A \emph{vertex self-avoiding random walk} is a sequence of vertices $(v_0,v_1,\dots)$ where no $v_i$ is repeated, i.e. a path without self intersections. One way to make this rigorous is to define an $n$-step vertex self-avoiding walk as a map $\omega: \{0,1,\dots,n\}\to \Z^d$ such that $|\omega(i+1)-\omega(i)|=1$ and $\omega(i)\neq \omega(j)$ for all $i\neq j$. Vertex self-avoiding random walks can be constructed as simple random walk with a constraint (i.e. conditioned on the event of having no self intersections). Because of this constraint, vertex self-avoiding random walks are not Markovian. As with the basic walk, traps are possible in this model, since a path can terminate in a vertex whose neighbors have all been visited before. A nice exposition of the theory of vertex self-avoiding random walks is given in \cite{self2}.

P\'{o}lya's original proof of Theorem \ref{polya} proceeded by counting the number of walks of length $n$ which were recurrent, dividing this by the total number of walks of length $n$, and taking the limit as $n\to \infty$. There were many reasons for introducing self-avoiding random walks--most notably their connection to chain-like entities in chemistry (\cite{self4})--but one of the most interesting mathematically is that for self-avoiding random walks the problem of counting paths is highly nontrivial. One of the driving questions in the field is how fast the set of non-intersecting paths of length $n$ grows in various graphs. Another interesting reason to study self-avoiding random walks is for their connection to load balancing and cover time. If the walk is forbidden to reuse any vertex, then this is optimal for load balancing. However, this restriction may make it impossible to cover all of the graph, which is far from optimal for cover time.

Because the vertex self-avoiding walk is disallowed from returning to the origin, one cannot frame the recurrence question in exactly the same way as for the simple random walk. If the individual vertex self-avoiding random walks are the object of study, then none are recurrent. One could attempt to define a transience-cycling dichotomy as we did in Section 1.\ref{sec:basic}, i.e. define a vertex self-avoiding walk to be \emph{transient} if the probability of the walk getting arbitrarily far away from the origin is nonzero, but it is unclear what the correct measure should be. The main problem here is that a simple random walk will almost surely have intersections, so one cannot get a useful measure on the set of infinite self-avoiding paths by conditioning infinite simple random walks. Furthermore, the set of trapped paths on $\Z^d$ is countable, but the set of paths which escape to infinity is uncountable when $d>1$. Thus, the question of cycling vs. transience is not the right question for the vertex self-avoiding random walk.


If one wishes to study vertex self-avoiding walks using simple random walks, then one must know the probability of two simple random walks intersecting. On $\Z^2$, with probability one, any two simple random walks $(x_0,x_1,\dots)$ and $(y_0,y_1,\dots)$ will have infinitely many $n$ with $x_n=y_n$. This is because the difference sequence $(x_0-y_0,x_1-y_1,\dots)$ is a simple random walk and so must hit $0$ infinitely many times by P\'{o}lya's Theorem. It is known that on $\Z^d$ for $d>4$ the probability of two independent simple random walks of length $n$ intersecting is bounded away from 0 as $n\to \infty$ (see \cite{self1}). Thus, vertex self-avoiding random walks are much easier to understand in high dimension than in low dimension because in high dimension there is enough space to avoid traps. In particular, it is shown in \cite{self3} that for $d>4$ the vertex self-avoiding random walk behaves like the simple random walk in the sense that it weakly converges to Brownian motion. The situation is much more complicated for $d=3$ and $d=4$, and is an active area of research.

There is another way in which a random walk can avoid itself. An \emph{edge self-avoiding random walk} is a sequence of vertices $(v_0,v_1,\dots)$ where no edge $(v_i,v_{i+1})$ is repeated, but vertices may be repeated. From the facts above it is easy to see that on $\Z^2$, with probability one, any two simple random walks will share an edge $(x_n,x_{n+1})= (y_n,y_{n+1})$ infinitely many times, simply because there are infinitely many $x_n=y_n$ and a nonzero probability each time that $x_{n+1}=y_{n+1}$. Similarly, on $\Z^d$ for $d>4$ the probability of two independent simple random walks of length $n$ sharing an edge is bounded away from 0 as $n\to \infty$ since sharing an edge requires sharing vertices. Again, the situation for $d=3$ and $d=4$ is much more complicated.

In the random basic walk, the path of the robot is not a simple random walk but behaves like one at vertices which have never been visited before, since all outgoing arcs are equally likely. At vertices which have been visited before, the random basic walk can act in two different ways. If the robot enters by the same port which was used to enter this vertex previously, then the robot will be in a cycle, i.e. the rest of the walk is completely determined. If the robot enters by a new label, then it cannot exit by the same arc it exited by previously, so the robot moves as a particle in an edge self-avoiding random walk. The facts above might lead one to the intuition that for the random basic walk there should be some critical dimension $d$ such that one expects the robot to get stuck in a cycle on $\Z^s$ for $s<d$ but one expects the robot to escape on $\Z^t$ for $t>d$. We will see in Chapter \ref{ch:infinite} that this intuition is incorrect.

We end this section with a note on how the random basic walk can also be understood to be analogous to a vertex self-avoiding random walk. Given a basic walk on a graph $G$ one can define a new walk, which will contain exactly the same information. Consider the directed graph $\hat{G}$ whose vertices are the arcs from $G$, and where we draw an arc from $e_1$ to $e_2$ if the target $v$ of $e_1$ is the source of $e_2$ and if the label $L(e_2) = (L(e_1) \mod{d(v)})+1$. Label this arc in $\hat{G}$ by $L(e_2)$. A basic walk on $G$ immediately defines a walk on $\hat{G}$. If a vertex of $\hat{G}$ is ever visited twice, this means a directed arc in $G$ is traversed twice, i.e. there is a cycle in $G$. Thus, until the random basic walk cycles, the walk on $\hat{G}$ behaves like a vertex self-avoiding random walk.


\section{Reinforced Random Walks}

Because self-avoiding random walks do not have a good recurrence-transience problem, one might consider a random walk which simply makes reusing a vertex or edge less likely rather than completely disallowed. This leads to the notion of reinforced random walks: a continuous family of probability distributions (depending on a parameter $\beta$) such that self-avoiding walks are obtained in the limit as $\beta\to 0$. Because no edge in a reinforced random walk is ever disallowed from use, there is no notion of trapping for these types of random walks. Unfortunately, reinforced random walks appear to be even harder to understand than self-avoiding random walks. We now discuss reinforced random walks in more detail.

The \emph{edge reinforced random walk}, originally introduced in \cite{edgereinforced}, is defined similarly to the simple random walk, but the probability of moving from a vertex $v_n$ to a neighbor $v_{n+1}$ depends on how many times each edge out of $v_n$ has been crossed. Let $G_n$ be the $\sigma$-field $\sigma(v_0,\dots,v_n)$, i.e. the history of the walk up to vertex $v_n$. Let $N(x,y)$ denote the number of times the edge $\{x,y\}$ has been traversed (in either direction) at the moment when the walker is at $v_n$. Then

$$P(v_{n+1} = w\;|\; G_n) = \frac{1+N(w,v_n)}{\sum_{y\sim v_n}(1+N(y,v_n))}$$

Unpacking this definition, we see that the edge reinforced random walker prefers edges it has walked on before. Similarly, the \emph{vertex reinforced random walk} (first introduced in \cite{vertexreinforced}) makes use of conditional probabilities, but weights $P(v_{n+1}=w\;|\; G_n)$ by the number of times the vertices neighboring $v_n$ have been visited. This number can be different from the number of times $\{v_n,z\}$ has been crossed, since the vertex $z$ could have been visited via edges which did not involve $v_n$. A comprehensive treatment of these two processes is given in \cite{reinforced}, and that is where the definition above comes from.

The recurrence question for vertex reinforced random walks is partially resolved in \cite{volkov}, where it is shown that for almost all graphs the walk will visit only finitely many vertices with positive probability. Furthermore, for all trees of bounded degree the walk will visit only finitely many vertices with probability 1. This result is conjectured for all graphs of bounded degree, but is not proven.

The recurrence question for edge reinforced random walks is settled for trees in Theorem 5.2 of \cite{reinforced}. However, it is still open for most graphs and in particular for all lattices $\Z^d$. This unfortunate gap in knowledge makes it impossible to apply the theory of edge reinforced random walks to random basic walks.

A simpler version of the edge reinforced random walk is the \emph{once-reinforced random walk}, where we only retain the information of whether an edge has been crossed or not, rather than the information of how many times it has been crossed. One change to the model from \cite{edgereinforced} is to allow a parameter $\beta>0$ governing how much the walker cares about traversing an edge it has already traversed. In this model, we allow our graph edges to have weights, and we begin with all edges weighted 1. Every time an edge with weight 1 is crossed the weight is changed to $\beta$. The probability that the walker will cross an edge $e$ out of a vertex $v$ is

$$\frac{weight(e)}{\sum_{w\sim v} weight(\{w,v\})}$$

Note that each of these weights depends on the walk leading up to $v$, just as in the model from \cite{edgereinforced}. If $\beta<1$ then the walker prefers to move along edges it has never traveled before. If $\beta>1$ then the walker prefers to retrace its steps. It is clear that in the limit as $\beta \to 0$ we recover the self-avoiding random walk. The question of recurrence vs. transience on once-reinforced random walks is addressed in \cite{selkie} and the walk is proven to be recurrent on $\Z$ and on $\{0,1\} \times \Z$ for all $\beta$. The question is still open on $\{0,1,2\} \times \Z$ and on $\Z^d$ for $d>1$.

There are numerous other generalizations of random walks in this vein, including reinforced random walks of sequence type, reinforced random walk of matrix type, vertex reinforced jump processes, weakly reinforced random walks, weakly self-avoiding walks, loop erased random walks, multi-agent random walks, random walks with restarts, and directionally reinforced random walks, and recurrence questions have been studied on all of them.

\section{Quasirandom processes}
\label{sec:quasirandom}

The second main trait of the random basic walk, after its ability to get trapped, is the fact that there are pointers on all vertices and a local deterministic rule by which the robot decides which direction to move next. This is reminiscent of the rotor router model popularized in recent years by Jim Propp, and we now discuss this model. Because the rotor router model is a model of quasirandom processes, we must first define some very general terms.

A \emph{random process} is a sequence of random variables describing a processes whose outcomes are controlled by probability distributions rather than a deterministic pattern, e.g. simple random walks are random processes describing the motion of the walker. A \emph{pseudorandom process }is a deterministically generated process whose behavior is designed to exhibit statistical randomness, e.g. random number generators. A \emph{quasirandom analogue of a random process} $X$ is a deterministic process designed to give the same limiting behavior as $X$ but with faster convergence. Quasirandom processes usually fail to exhibit statistical randomness but do capture information like recurrence. The quasirandom analogue of simple random walks on directed graphs is the \emph{rotor router model}.

The rotor router model consists of the following data. At every vertex $v$ fix a cyclic ordering on the outgoing arcs $e_1,e_2,\dots,e_d$ where $d$ is the out-degree of $v$. From this ordering, create an infinite periodic sequence of period $d$ via $e_1,e_2,\dots,e_d,e_1,e_2,\dots)$. This sequence is called the \emph{rotor pattern} at $v$. Next, a \emph{rotor} is placed at each $v$--initially pointing in the direction $e_1$--and a particle is started at some initial vertex $v_0$. The rotor router model can be understood as an automaton, with states $(v,\ell)$ where $v\in G$ is the position of the particle and $\ell$ is the direction of the rotor at $v$. The transition rule is as follows: when $v$ is first visited, the rotor rotates to $e_2$ and then the particle exits by $e_1$. When $v$ is next visited, the rotor rotates to $e_3$ and then the particle exits by $e_2$, etc. Thus, the direction of the rotor determines which vertex a particle visiting $v$ visits next. The collection of rotor patterns over all vertices is called the \emph{rotor configuration}. The path of the particle is called a \emph{rotor walk}.

Rotor routers were first introduced in \cite{rotor1} as an offshoot of the abelian sandpile model of \cite{rotor2} which studies self-organized criticality. They have been heavily studied by Jim Propp (e.g. \cite{rotor5, rotor3, rotor4}) and are claimed to ``better than random'' because the central limit theorem behavior is achieved immediately, rather than requiring a large number of time steps. For instance, if one iterates a rotor router on $\Z$ starting at any vertex $0 < v < n$ and stopping when either $0$ or $n$ is reached, then exactly half the particles will reach $0$ before $n$. In particular, subsequent runs will alternate which of these two options occurs, and this holds for any rotor configuration. The Central Limit Theorem predicts this behavior, but would require many more time steps to conclude that the proportion of particles reaching 0 before $n$ is $1/2$. Similarly, for any target vertex $t$, iterated rotor walks will hit $t$ with the same frequencies as one would expect from a simple random walk.

Rotor routers have found numerous applications, including load balancing for parallel processing, protocol broadcasts on networks, mergesort, and internal diffusion limited aggregation (\cite{rotor8}). A nice reference for facts about rotor routers is \cite{rotor6}. Unlike the random basic walk or the simple random walk, a rotor walk is completely deterministic. It does not appear that any papers have studied what happens when the rotor patterns are chosen according to a random process as is done in the random basic walk. We now pause to formalize this idea.

Let $G$ be a directed graph. For every vertex $v$, let $C_v$ denote the set of cyclic orderings on the set of arcs going out of $v$. Let $\mathcal{X}_v$ be a probability distribution on this set and choose a cyclic ordering based on $\mathcal{X}_v$. The rotor pattern selected is the \emph{random rotor pattern }at $v$ with respect to $\mathcal{X}_v$. The collection of random rotor patterns over all $v\in G$ is the \emph{random rotor configuration} on $G$ with respect to $\{\mathcal{X}_v \;|\; v\in G\}$. Placing a rotor on each vertex, selecting an initial vertex $v_0$, and starting a particle from $v_0$ with the transition rule of the rotor rotor model defines the \emph{random rotor router model} as an automaton with states $(v,\ell)$ just like the rotor router model. The path of the particle is called a \emph{random rotor path}.


The random rotor router model is most analogous to the random basic walk when all the probability distributions above are uniform distributions. Because this notion has not appeared in the literature, very little is known about the random rotor router model, and there are numerous interesting questions. Some results are known (the theorem below, for instance) because some of the results in the literature for rotor routers do not depend on the rotor configuration. However, many of the papers on rotor routers require the rotor configuration to satisfy certain properties, and it would be interesting to know whether the results of these papers are also true for random rotor routers.

There are several other differences between rotor walks (even random rotor walks) and the random basic walk. For instance, no matter which direction a vertex $v$ is entered from, the rotor will determine the exit direction. This leads to rotor walks bouncing around between a small number of vertices until rotors have rotated sufficiently to allow the walker to escape the cluster. Furthermore,  the rotor walk can never become trapped. This is not immediately clear, but is a consequence of the following theorem of Cooper and Spencer (\cite{rotor7}):

\begin{thm}[Cooper-Spencer Theorem] Suppose one begins with a set of particles on vertices of $\Z^d$ with even distance from the origin, and that all particles undergo simultaneous rotor walks for $n$ steps. Let $RR(v)$ denote the number of particles at $v$ at the conclusion of this process. Let $RW(v)$ denote the expected number of particles at $v$ if the particles moved via simple random walks rather than rotor walks. Then there is a constant $C_d$ depending on $d$ but not on the rotor configurations, the time $n$, or the initial placement of particles such that for all $v$

$$|RR(v)-RW(v)| < C_d$$

Furthermore, the constant $C_d$ is known for $d=1$ and $d=2$
\end{thm}

This theorem also answers the recurrence question for rotor routers and random rotor routers. In particular, it says that if we start a particle at the origin and allow a rotor walk to run for infinitely many steps then in $\Z^2$ we expect all vertices to be visited but in $\Z^d$ for $d>2$ it is possible some vertices will not be visited. Thus, the (random) rotor walk behaves in a recurrent manner on $\Z^2$ and in a transient manner on $\Z^d$ for $d>2$, though the notions of recurrence and transience of rotor routers do not appear to be in the literature. This result is another difference between the (random) rotor walk and the random basic walk, as can be seen from Theorem \ref{lattices}.

\chapter{Examples}
\label{sec:examples}

In this chapter we will consider several examples of the basic walk, both to gain a better understanding of how it works and because these examples will be used again in the proofs of Chapter \ref{ch:infinite}. The starting vertex will always be denoted by $v$, and the port orientations are given in the figures. If there is a label $\ell$ on an arc into $v$ then the reader should start the basic walk from $v$ with the label $\ell+1$ and follow arcs of increasing port number step by step until the pattern is understood. If there is no label, then the reader should start with the arc leaving $v$ with label $1$, and should try different initial labels until they see the pattern.

\section{Trapping Configurations}
\begin{example} The following port orientations near vertex $v \in \Z^2$ force the robot to be trapped:
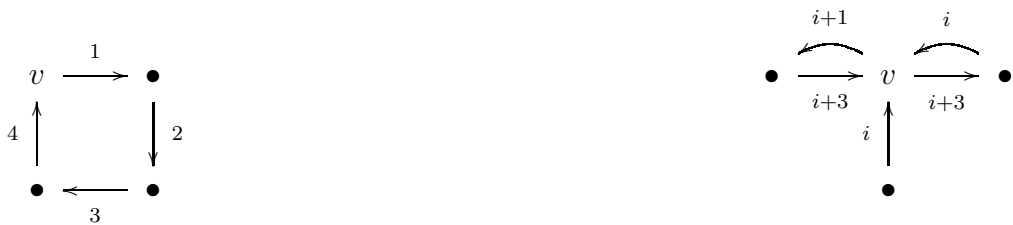
\label{latticetraps}

\xymatrix{v \ar[r]^1 & \bullet \ar[d]^2 & & & & & & \bullet \ar[r]_{i+3} & v \ar@/_1pc/[l]_{i+1} \ar[r]_{i+3}& \bullet \ar@/_1pc/[l]_i\\\bullet \ar[u]^4 & \bullet \ar[l]^3 & & & & & & & \bullet \ar[u]^i & }

The second can be rotated to yield four configurations, where $v$ is entered either from the north, south, east, or west. We will refer to any of these four port orientations (with any entering port $i$) as a \emph{trapping configuration $T$}. In a random basic walk, the configuration $T$ occurs with probability $c=(1/4)^3\cdot(1/3)$ since the probability of the edge labeled $i+1$ receiving that label is $1/4$, the probability of $i+2$ is $1/4$, the probability of $i+3$ is $1/3$, and the probability of $i$ is $1/4$. This configuration $T$ and the fact that it occurs with probability $c>0$ will be used in the proof of Theorem \ref{2d}.
\end{example}

It is easy to generalize $T$ to hold in $\Z^d$, using the fact that $\Z^d$ is $2d$-regular so a node only needs $d$ neighbors to form a trap. Rather than 3 collinear vertices, the trap will be made up of a central vertex $v$ and $d$ neighbors all living on some hyperplane of dimension $d-1$ in $\Z^d$. Call such configurations $T_d$. They exist because the hyperplane is $2(d-1)$-regular and $2d-2 \geq d$ for all $d\geq 2$. These examples will be used in Chapter \ref{ch:infinite} to understand how random basic walks on $\Z^d$ behave.

Graphs other than $\Z^d$ can also have traps. For instance, let $G$ be the hexagonal lattice. Often when considering questions on infinite lattices it is easier to work with the hexagonal lattice, but the random basic walk is one interesting place where it is harder. The reason is that the hexagonal lattice is 3-regular but triangle free. This means that a cycle must use at least 6 arcs, so a trap at a vertex $v$ cannot consist entirely of neighbors of $v$ which have not been visited by the robot yet.

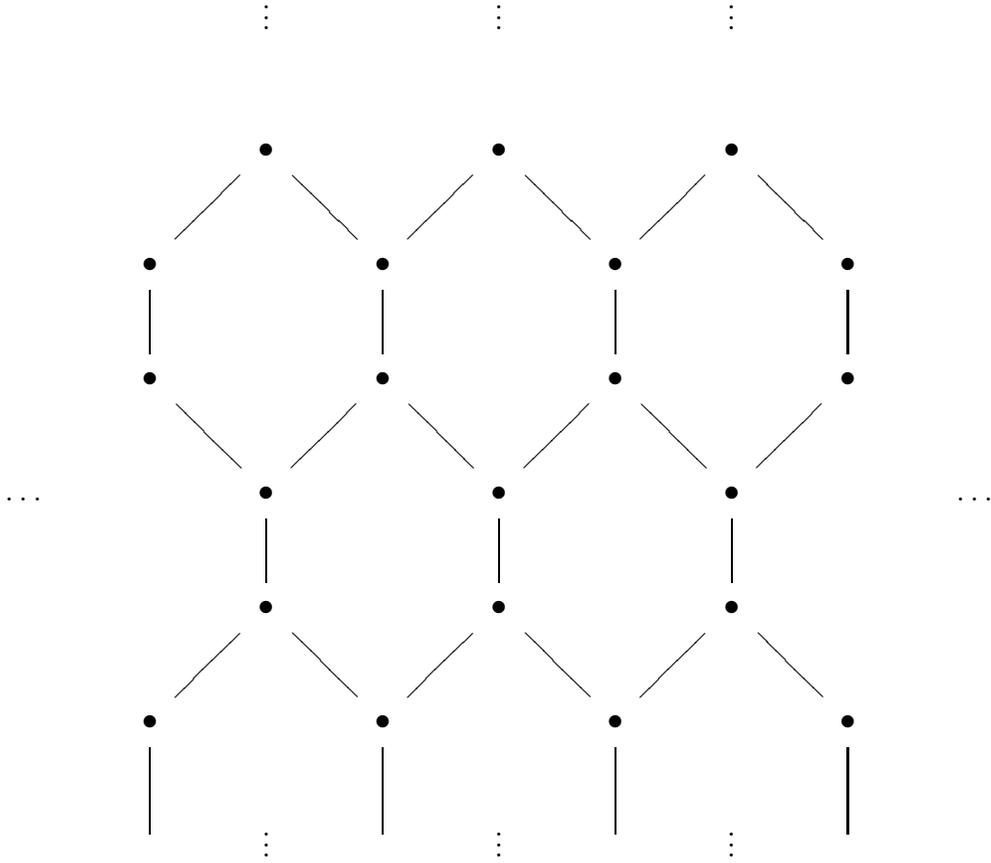
\xymatrix{&&\vdots&&\vdots&&\vdots&&\\
&&\bullet \ar@{-}[dr] \ar@{-}[dl] & & \bullet \ar@{-}[dr] \ar@{-}[dl]&&\bullet \ar@{-}[dr] \ar@{-}[dl]&&\\
&\bullet \ar@{-}[d]&&\bullet \ar@{-}[d]&&\bullet \ar@{-}[d]&&\bullet \ar@{-}[d]&\\
&\bullet \ar@{-}[dr] && \bullet \ar@{-}[dr] \ar@{-}[dl]&&\bullet \ar@{-}[dr] \ar@{-}[dl]&&\bullet \ar@{-}[dl]&\\
\dots&&\bullet \ar@{-}[d] && \bullet \ar@{-}[d]&&\bullet \ar@{-}[d]&&\dots\\
&&\bullet \ar@{-}[dl] \ar@{-}[dr]&& \bullet \ar@{-}[dl]\ar@{-}[dr]&&\bullet \ar@{-}[dr]\ar@{-}[dl]&&\\
&\bullet \ar@{-}[d]&&\bullet \ar@{-}[d]&&\bullet \ar@{-}[d]&&\bullet \ar@{-}[d]&\\
&&\vdots&&\vdots&&\vdots&&}

\begin{example} \label{graphtraps} The following configurations can trap the robot on the hexagonal lattice above.

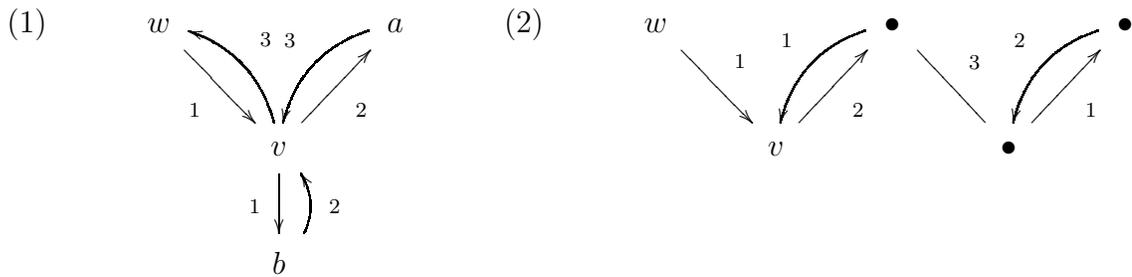
\xymatrix{(1) &w \ar[dr]_1 & & a \ar@/_1pc/[dl]_3 &(2)& w\ar[dr]^1 & & \bullet \ar@{-}[dr]^3 \ar@/_1pc/[dl]_1& & \bullet \ar@/_1pc/[dl]_2\\
&& v \ar[ur]_2 \ar[d]_1 \ar@/_1pc/[ul]_3 & && & v \ar[ur]_2 & & \bullet \ar[ur]_1 & \\
&& b \ar@/_1pc/[u]_2 &&&&&&&}
\end{example}

The left-hand trap is a star-shaped trap which consists of the central vertex $v$, two neighbors $a$ and $b$ which have not been visited yet by the robot, and the neighbor $w$ from which the robot entered $v$. This trap requires the use of a vertex the robot has visited before, however only the arc from $v$ back to $w$ is needed for the trap to exist, so the existence of the trap is independent from the walk leading up to $v$, as long as $v,a,$ and $b$ have not been visited yet. We refer to this trap as a \emph{star} $V$.

The right-hand trap consists of three vertices which have not been visited before and a path between them which will trap the robot forever between these three vertices and $v$. We refer to this trap as a \emph{spire} $S$. In a random basic walk, both $V$ and $S$ occur with constant probability.

\section{Non-regular graphs}

On a regular graph $G$ of degree $d$, one knows that the basic walk will take a step labeled $t\mod{d}$ in time step $t$, i.e. there is a global clock controlling which label to use. This is not true on non-regular graphs:

\begin{example} \label{stargraph} Consider the star graph, where the port numbers shown are for $v$, and all port numbers coming into $v$ are $1$, since all nodes other than $v$ have degree 1.
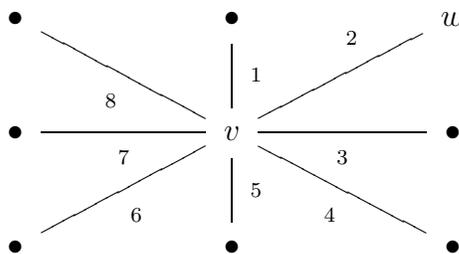
\xymatrix{\bullet & & \bullet& & w \\
\bullet& & v \ar@{-}[urr]^(.6)2 \ar@{-}[u]_1 \ar@{-}[rr]_3 \ar@{-}[drr]_4 \ar@{-}[d]^5 \ar@{-}[dll]^6 \ar@{-}[ll]^7 \ar@{-}[ull]^8 & & \bullet \\
\bullet & & \bullet& & \bullet}
\end{example}

If the basic walk was controlled by a global clock then this would suggest the robot follows the 1 arc out of $v$, then returns, then follows the 3 arc, then returns, then follows the 5 arc, etc. However, what actually occurs is that no matter which arc the robot follows out of $v$ it returns by an arc labeled 1. Thus, it must leave $v$ by the arc labeled 2. It is now trapped between that vertex $w$ and $v$ forever.

\section{Interesting labelings on $\Z^2$}
\label{integerlabelings}

A nice property of simple random walks is that they are guaranteed to visit every vertex in $\Z^2$ with probability 1 if left to run for enough steps. It will be shown in Chapter \ref{ch:infinite} that this property fails for the random basic walk. However, given the initial vertex and port number in $\Z^2$ it is possible to produce infinitely many labelings such that the basic walk will explore the whole graph.

\begin{example} \label{spiral} Fix an initial vertex $v_0$ and an initial direction $i$. Then there is an infinite family of labelings on $\Z^2$ starting from $v_0$ and direction $i$ which explore the whole plane. This is obtained by spirals with east-west stretch factor $k\in \Z_{>0}$. The cases shown are $k=1$ and $k=2$:

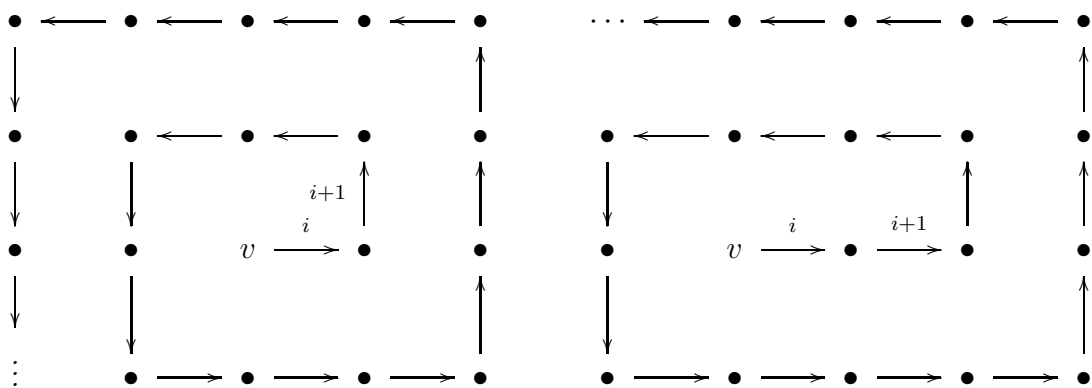
\xymatrix{
\bullet \ar[d] & \bullet \ar[l] & \bullet \ar[l] & \bullet \ar[l] & \bullet \ar[l] &  \cdots & \bullet \ar[l] & \bullet \ar[l] & \bullet \ar[l] & \bullet \ar[l]\\
\bullet \ar[d] & \bullet \ar[d] & \bullet \ar[l] & \bullet \ar[l] & \bullet \ar[u] & \bullet \ar[d] & \bullet \ar[l] & \bullet \ar[l] & \bullet \ar[l] & \bullet \ar[u]\\
\bullet \ar[d] & \bullet \ar[d] & v \ar[r]^i & \bullet \ar[u]^{i+1} & \bullet \ar[u] &  \bullet \ar[d] & v \ar[r]^i & \bullet \ar[r]^{i+1} & \bullet \ar[u] & \bullet \ar[u]\\
\vdots & \bullet \ar[r] & \bullet \ar[r] & \bullet \ar[r] & \bullet \ar[u] &  \bullet \ar[r] & \bullet \ar[r] & \bullet \ar[r] & \bullet \ar[r] & \bullet \ar[u]\\
}
\end{example}

One should avoid the assumption that the starting vertex and port are fixed before the labeling is chosen, since in practice a robot could be placed anywhere by its owner and expected to explore its surroundings. Without the assumption of a fixed starting place and initial port, the author does not know if there is a labeling such that the basic walk visits every vertex. However, there are labelings such that the basic walk from any vertex and with any initial port is transient. In this figure, each edge represents arcs going in both directions, and the label on the edge is the label on both of the two arcs.

\begin{example} \label{transient} A labeling in $\Z^2$ where any basic walk escapes to infinity:

\xymatrix{
     & & & \vdots & & & \\
     & \bullet \ar@{-}[l]^2 \ar@{-}[u]^1 \ar@{-}[d]^3 \ar@{-}[r]^4 & \bullet \ar@{-}[u]^1 \ar@{-}[d]^3 \ar@{-}[r]^2 & \bullet \ar@{-}[u]^1 \ar@{-}[d]^3 \ar@{-}[r]^4 & \bullet\ar@{-}[u]^1 \ar@{-}[d]^3 \ar@{-}[r]^2 &\bullet \ar@{-}[d]^3 \ar@{-}[u]^1 \ar@{-}[r]^4 &\\
     &\bullet\ar@{-}[l]^2 \ar@{-}[d]^1 \ar@{-}[r]^4 &\bullet\ar@{-}[d]^1 \ar@{-}[r]^2 &\bullet\ar@{-}[d]^1 \ar@{-}[r]^4 &\bullet\ar@{-}[d]^1 \ar@{-}[r]^2 &\bullet \ar@{-}[d]^1 \ar@{-}[r]^4 &\\
\dots &\bullet \ar@{-}[l]^2 \ar@{-}[d]^3 \ar@{-}[r]^4 &\bullet \ar@{-}[d]^3 \ar@{-}[r]^2 &\bullet \ar@{-}[r]^4 \ar@{-}[d]^3 &\bullet \ar@{-}[r]^2 \ar@{-}[d]^3 &\bullet \ar@{-}[r]^4 \ar@{-}[d]^3 & \dots \\
     &\bullet\ar@{-}[l]^2 \ar@{-}[d]^1 \ar@{-}[r]^4 &\bullet\ar@{-}[d]^1 \ar@{-}[r]^2 &\bullet \ar@{-}[d]^1 \ar@{-}[r]^4 &\bullet \ar@{-}[d]^1 \ar@{-}[r]^2 &\bullet \ar@{-}[d]^1 \ar@{-}[r]^4 & \\
     &\bullet \ar@{-}[l]^2 \ar@{-}[r]^4 \ar@{-}[d]^3&\bullet\ar@{-}[r]^2 \ar@{-}[d]^3 &\bullet \ar@{-}[r]^4 \ar@{-}[d]^3 &\bullet \ar@{-}[r]^2 \ar@{-}[d]^3 &\bullet \ar@{-}[r]^4 \ar@{-}[d]^3 & \\
 & & & \vdots & & &
}
\end{example}

Every starting vertex $v$ and port $i$ leads to an infinite staircase moving in the directions specified by $i$ and $i+1$ which gets further away from $v$ with every step. This example generalizes to $\Z^d$ with a staircase consisting of a sequence of moves, one in each of the $d$ directions. This example yields an infinite family of examples where the basic walk escapes to infinity from any starting vertex and any initial label: simply add in blocks of 4 columns which act like plateaus for the robot to move east or west for $4n$ steps between a given north-south step on the staircase. The case above is $n=0$; the case below is $n=1$, with dotted lines to distinguish the block of new columns. For $n>1$ one must simply insert $n$ copies of this block:

\begin{example} \label{transient2} An infinite family of labelings where all basic walks escape to infinity:

\xymatrix{
    & & & & \vdots & & & & & \\
    &\bullet \ar@{-}[l]^4 \ar@{-}[d]^3 \ar@{-}[u]_1& \bullet \ar@{-}[l]^2 \ar@{..}[u]^1 \ar@{..}[d]^3 \ar@{-}[r]^4 & \bullet \ar@{-}[u]^3 \ar@{-}[d]_2 \ar[r]^1 & \ar@/_1pc/[l]_3 \bullet \ar@{-}[u]^1 \ar@{-}[d]^4 \ar@{-}[r]^2 & \bullet\ar@{-}[u]^4 \ar@{-}[d]_1 \ar[r]^3 &\ar@/_1pc/[l]_1\bullet \ar@{-}[d]^2 \ar@{-}[u]_3 \ar@{-}[r]^4 &\bullet \ar@{..}[d]_3 \ar@{..}[u]^1 \ar@{-}[r]^2 & \bullet \ar@{-}[r]_4 \ar@{-}[d]_3 \ar@{-}[u]^1&\\
    & \bullet \ar@{-}[l]^4 \ar@{-}[d]^1 &\bullet\ar@{-}[l]^2 \ar@{..}[d]^1 \ar@{-}[r]^4 &\bullet\ar@{-}[d]_3 \ar[r]^1 &\ar@/_1pc/[l]_3\bullet\ar@{-}[d]^1 \ar@{-}[r]^2 &\bullet\ar@{-}[d]_4 \ar[r]^3 &\ar@/_1pc/[l]_1\bullet \ar@{-}[d]^3 \ar@{-}[r]^4 &\bullet \ar@{..}[d]_1 \ar@{-}[r]^2 &\bullet \ar@{-}[r]_4 \ar@{-}[d]_1 &\\
\dots & \bullet \ar@{-}[l]^4 \ar@{-}[d]^3&\bullet \ar@{-}[l]^2 \ar@{..}[d]^3 \ar@{-}[r]^4 &\bullet \ar@{-}[d]_2 \ar[r]^1 &\ar@/_1pc/[l]_3\bullet \ar@{-}[r]^2 \ar@{-}[d]^4 &\bullet \ar[r]^3 \ar@{-}[d]_1 &\ar@/_1pc/[l]_1\bullet \ar@{-}[r]^4 \ar@{-}[d]^2 &\bullet \ar@{..}[d]_3 \ar@{-}[r]^2 & \bullet \ar@{-}[r]_4 \ar@{-}[d]_3 & \dots \\
    & \bullet \ar@{-}[l]^4 \ar@{-}[d]^1 &\bullet\ar@{-}[l]^2 \ar@{..}[d]^1 \ar@{-}[r]^4 &\bullet\ar@{-}[d]_3 \ar[r]^1 &\ar@/_1pc/[l]_3\bullet \ar@{-}[d]^1 \ar@{-}[r]^2 &\bullet \ar@{-}[d]_4 \ar[r]^3 &\ar@/_1pc/[l]_1\bullet \ar@{-}[d]^3 \ar@{-}[r]^4 &\bullet \ar@{..}[d]_1 \ar@{-}[r]^2  & \bullet \ar@{-}[r]_4 \ar@{-}[d]_1 &\\
    &\bullet \ar@{-}[l]^4 \ar@{-}[d]^3 &\bullet \ar@{-}[l]^2 \ar@{-}[r]^4 \ar@{..}[d]^3&\bullet\ar[r]^1 \ar@{-}[d]_2 &\ar@/_1pc/[l]_3\bullet \ar@{-}[r]^2 \ar@{-}[d]^4 &\bullet \ar[r]^3 \ar@{-}[d]_1 &\ar@/_1pc/[l]_1\bullet \ar@{-}[r]^4 \ar@{-}[d]^2 &\bullet \ar@{..}[d]_3 \ar@{-}[r]^2  & \bullet \ar@{-}[r]_4 \ar@{-}[d]_3 & \\
 & & & & \vdots & & & & &
}
\end{example}

A basic walk which begins outside the block will either move away from it by a staircase or will move towards and eventually connect with the new columns. At that point, the robot will enter by a port 4 and will cross without moving north or south. On the other end, the staircase will continue and the robot will escape. A basic walk which begins inside will always exit and then move away via a staircase. If the initial port used takes the robot east or west, then the robot will cross through the block and move away by a staircase. If the initial port used takes the robot north or south, then it is clear that the robot can move by at most two rows before turning and exiting on an east-west path. This example demonstrates that Theorem \ref{lattices} only holds for almost all labelings.

\chapter{Random Basic Walks on Infinite Graphs}
\label{ch:infinite}

This chapter addresses the question of transience vs. cycling on infinite graphs $G$, where the assumptions of connectedness, countability, and local finiteness are always implicit. The question is resolved for lattices $\Z^d$ in Section \ref{sec:lattices}, for regular graphs in Section \ref{sec:regular}, and for graphs of bounded degree in Section \ref{sec:locallyfinite}, where also an example is given to show the hypothesis of bounded degree cannot be dropped. In Section \ref{sec:regular}, a proposition is proven which shows that there are arbitrarily long paths and cycles on any regular graph, but the random basic walk is unlikely to hit them.
\section{Transience and Cycling on Integral Lattices}
\label{sec:lattices}

Recall the configurations $T$ from Example \ref{latticetraps}. Such $T$ could occur any time the random basic walk on $\Z^2$ visits the center vertex of three collinear vertices where none has been visited before. If any of the vertices have been visited before, then $T$ may be impossible, since it could be the case that the one of the labels $T$ requires has already been assigned to a different arc by a previous step of the robot. There are other ways for the robot to be trapped, but we need only consider a specific instance of $T$ for Theorem \ref{2d}. A transient random basic walk must avoid $T$ infinitely many times, but this is impossible, since $T$ occurs with constant, nonzero probability.

\begin{thm} \label{2d} In $\mathbb{Z}^2$, the random basic walk cycles with probability $1$.
\end{thm}

\begin{proof}[\textbf{Shells Method}] Let $v_0$ denote the starting vertex. We'll prove that the probability of cycling is 1, i.e. the probability of escaping to infinity is zero. Once we have this fact, the probability of a transient basic walk will be $\Pr(\exists$ vertex $x$ which the walk can escape to infinity from$)=0$ because it's a countable union of events, each of probability zero. Consider the following picture:

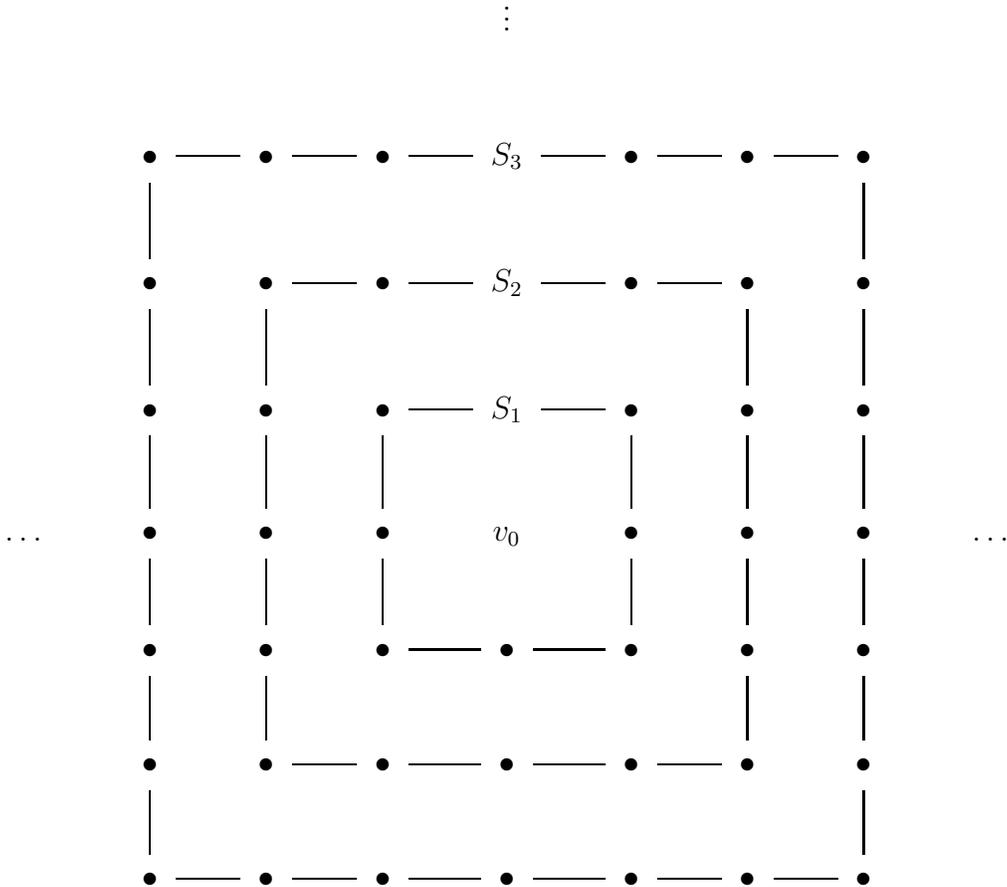
\xymatrix{
  & & & & \vdots & & & & \\
 & \ar@{-}[d] \ar@{-}[r] \bullet & \bullet \ar@{-}[r] & \ar@{-}[r] \bullet & S_3 \ar@{-}[r] & \ar@{-}[r] \bullet & \bullet \ar@{-}[r] & \ar@{-}[d] \bullet & \\
 & \ar@{-}[d] \bullet & \bullet \ar@{-}[r] \ar@{-}[d] & \ar@{-}[r] \bullet & S_2 \ar@{-}[r] & \bullet \ar@{-}[r] & \bullet \ar@{-}[d] & \bullet \ar@{-}[d]& \\
 & \ar@{-}[d] \bullet & \bullet \ar@{-}[d] & \ar@{-}[r] \ar@{-}[d] \bullet & \ar@{-}[r] S_1 & \ar@{-}[d] \bullet & \bullet \ar@{-}[d] & \bullet \ar@{-}[d]& \\
\dots & \ar@{-}[d] \bullet & \bullet\ar@{-}[d] & \ar@{-}[d] \bullet &    v_0      & \bullet \ar@{-}[d] & \ar@{-}[d] \bullet & \bullet \ar@{-}[d]& \dots \\
 & \ar@{-}[d] \bullet & \bullet \ar@{-}[d] & \ar@{-}[r] \bullet & \bullet \ar@{-}[r] & \bullet & \bullet \ar@{-}[d] & \bullet \ar@{-}[d]& \\
 &\ar@{-}[d] \bullet & \bullet\ar@{-}[r] & \ar@{-}[r] \bullet & \bullet \ar@{-}[r] & \ar@{-}[r] \bullet & \bullet & \bullet \ar@{-}[d]& \\
 &\ar@{-}[r] \bullet &\bullet \ar@{-}[r] &\bullet \ar@{-}[r] &\bullet \ar@{-}[r] & \bullet \ar@{-}[r] & \bullet \ar@{-}[r]& \bullet & \\
}

Here the $S_n$ are concentric squares (a.k.a. shells) of side length $2n$ centered at $v_0$. For the random basic walk to be transient, the robot must pass $S_n$ for all $n$, since any basic walk on a finite graph trivially cycles. If the robot reaches $S_n$, then denote by $v_n$ the first vertex reached on $S_n$. Because $S_n$ is a square, $v_n$ cannot be a corner as there is no arc from the interior of $S_n$ to a corner. This means $v_n$ is the center of three collinear vertices on $S_n$, none of which have been visited, i.e. there is a constant, non-zero probability $c$ that the configuration $T$ will occur at $v$. Indeed, $c = (1/4)^3\cdot (1/3)$ as discussed in Example \ref{latticetraps}. Let $E_n$ be the event that the walk reaches $S_n$ and the first time it does so (i.e. at $v_n$) is not a trapping configuration.


Let $E$ be the event that the robot gets infinitely far from $v_0$, i.e. the basic walk is transient. Clearly $E = \bigcap E_n$ because in order to get distance greater than $n$ from $v_0$ the robot must pass $S_n$ and must not be caught in a trap. For the basic walk to be transient, this must occur for all $n$. Conversely, if the robot passes each $S_n$ then the walk is clearly transient.

Note that $E_1\supset E_2 \supset E_3\dots$, since the robot reaching $S_n$ implies it reached $S_{n-1}$ (and didn't get trapped) but there are other ways to reach $S_{n-1}$. This means $\Pr(E) = \prod_{n\in \N}{\Pr(E_n \;|\; E_{n-1})}$. To bound $\Pr(E_n \;|\; E_{n-1})$ note that the probability of a path from $S_{n-1}$ to $S_n$ existing is $\leq 1$ and the probability that the first vertex visited on $S_n$ is not a trap is $\leq 1-c$, which is strictly less than 1 because $c>0$. Thus, $\Pr(E_n \;|\; E_{n-1})\leq 1*(1-c)$ for all $n$. This implies $\Pr(E)\leq \prod_{n\in \N}(1-c) = 0$.

\end{proof}

As usual for infinite random processes, the theorem only holds almost everywhere, as demonstrated by Example \ref{transient}. This theorem is not so surprising, since the analogous fact about recurrence on $\Z^2$ holds in the simple random walk and all its generalizations discussed in Chapter \ref{ch:history}. More surprising is that the same proof idea generalizes to show the following result:

\begin{thm}\label{lattices} For any $d$, the random basic walk on $\Z^d$ cycles with probability $1$.
\end{thm}

\begin{proof}
For $d=1$, it is easy to see that most labelings cause traps of size 2. Any time a pair of adjacent vertices has both arcs between them with non-equal labels, the robot will move back and forth between the two forever. Thus, in order for a random basic walk on $\Z$ to be transient, the labels must alternate $\dots,1,2,1,2,1,2,\dots$. This occurs with probability $0 = \lim_{n\to \infty} (1/2)^n$, so the random basic walk cycles with probability $1$.

For $d>2$, the proof uses the same ideas as in Theorem \ref{2d}, but with the trapping configurations $T_d$ instead of $T$ and with shells $S_n$ which are $d$-dimensional hypercubes of side length $2n$ centered at $v_0$. The configuration $T_d$ is defined such that the robot moves to a neighbor, then returns to $v$, then moves to a different neighbor, etc until $2d$ arcs have been traversed. At this point the port numbers at $v$ force the robot to retrace these same $2d$ steps ad infinitum. As above, let $v_n$ be the first vertex visited on $S_n$. The configurations $T_d$ can occur on any $d-1$ dimensional face of $S_n$ since the neighbors of $v_n$ will have no port numbers assigned when $v_n$ is first visited (as in the $d=2$ case).

The configuration $T_d$ occurs with constant probability $c_d=(1/2d)^d\cdot(1/2d)(1/(2d-1))\dots(1/(d+1))$ because of the $d$ port numbers at $v$ and the individual port number which must be assigned at each of the $d$ neighbors of $v$ used in the trap. Clearly $c_d>0$ so $1-c_d<1$. Define events $E_n$ and $E$ as in the $d=2$ case. The same proof (using $c_d$ instead of $c$) proves that $\Pr(E) = \prod_{n\in \N}(1-c_d) = 0$, i.e. the random basic walk on $\Z^d$ cycles.
\end{proof}


The fact that the random basic walk continues to cycle as $d$ grows, rather than becoming transient at some $d$, proves that the random basic walk is very different from types of random walks studied previously. The proof of Theorem \ref{lattices} will generalize to any class of graphs on which shells can be formed and on which the number of neighbors on a shell is larger than the number needed to form a trap, i.e. a class of graphs satisfying certain expander properties.

The method of proof above actually gives a bit more than just the statement of the theorem--it can be used to get a bound on the expected number of vertices visited by the random basic walk. For instance, in $\Z^2$ we have $c=1/192$. If the robot only moves in one direction, say east, then every step reaches a new $S_n$ and has a chance of yielding a trap. The event of trapping is then described by a geometric distribution, and the expected number of vertices visited before a trap is reached is 192. Including the trapping vertices boosts this overall number to 194.

The maximum number of vertices visited is obtained by the spiral of Example \ref{spiral}. With this labeling the robot could visit all vertices inside $S_n$ before reaching $S_{n+1}$. By the argument above the robot is expected to reach 192 shells, so this means the robot visits $192^2+2 = 36866$ vertices. Hence, $E($vertices visited$) \leq 36866$ for the random basic walk in $\Z^2$, and this bound also holds for the graphs $G_{k,n}$ for $k,n\geq 192$. Consideration of more ways to trap than just the configuration $T$ would boost the probability of trapping and thereby give a better bound than 36866.

Determining the expected number of vertices visited precisely appears to be a non-trivial problem. We will return to this question in Chapter \ref{conclusion}. We remark here that the question of the expected maximum length of a cycle is much simpler. Proposition \ref{prop:cyclelength} shows that there are cycles of arbitrarily large size which the robot can fall into, but in practice the random basic walk will most likely become trapped in a small cycle. Note that there is a labeling where every vertex is in a 4-cycle, but this labeling occurs with probability zero, so does not contradict Proposition \ref{prop:cyclelength}.

\section{Transience and Cycling on Regular Graphs}
\label{sec:regular}

Robots should be able to explore regions which are not square-shaped, so this leads to consideration of graphs other than $\Z^d$. For regular graphs of degree $d$, the random basic walk requires at least $d+1$ vertices to form a trap, due to the global clock mentioned in Example \ref{stargraph}. This global clock can tell any vertex which label to use based solely on how many steps have been taken so far. The reason is that all vertices can agree on how to get from a global time $t$ to a permissible label $i$--namely, set $i = t \mod d$ and if this value is 0 then set $i=d$.

Recall that the proof of Theorem \ref{lattices} generalizes to prove some class of expander graphs have cycling random basic walk. Unfortunately, this class of graphs does not contain all regular graphs. In particular, it does not contain the hexagonal lattice, due to the considerations in Example \ref{graphtraps} regarding the fact that this graph is triangle free. Thus, a new proof method is needed to prove the random basic walk cycles on the hexagonal lattice, and this proof method will generalize to arbitrary regular graphs (where $G$ is always assumed to be countable, connected, and locally finite).

\begin{thm} \label{hexagonal} On the hexagonal lattice, the random basic walk cycles with probability 1.
\end{thm}

\begin{proof}[\textbf{Spires Method}] Recall the spires from Example \ref{graphtraps}. Denote the starting location of the random basic walk by $v_0$. Let $S_n = \{v\;|\;d(v_0,v)=4n\}$ and let $D_n = \{v\;|\;d(v_0,v)\leq 4n\}$. Let $E_n$ be the event that the robot reaches $S_n$ and that the first vertex where that occurs (call it $v_n$) does not have a spire in the region between $D_n$ and $S_{n+1}$. Note that spires consist of 3 vertices past $v_n$, but the number of vertices between $D_n$ and $S_{n+1}$ is 4, so the spire will be completely contained in the region between $D_n$ and $S_{n+1}$. Spires in this region are guaranteed to consist of vertices which have not been visited by the robot before, i.e. the port orientations are independent of the history of the random basic walk.

Independence of port orientations guarantees that the probability of a spire from $v_n$ towards $S_{n+1}$ is a constant probability $c>0$. For the spire drawn in Example \ref{graphtraps}, the probability is $c=(1/3)^3(1/2)^3$ because of the six arcs used. The probability of getting from $S_n$ to $S_{n+1}$ without re-entering $D_n$ is $\leq 1$. Thus, $\Pr(E) = \prod_{n\in \N}{\Pr(E_n \;|\; E_{n-1})} \leq \prod_{n\in \N}{(1-c)} = 0$, proving that the random basic walk cycles with probability 1.
\end{proof}

It is likely that analogues for Example \ref{transient} exist for the hexagonal lattice, i.e. that it is possible to have a transient labeling. The author does not know if there is a way to place such a labeling on an arbitrary regular graph $G$.

Note that the proof above generalizes for any regular graph $G$ of degree $d>2$. The case for $d=2$ is analogous to the argument for $\Z$ in the previous section. For $d>2$, replace the $4n$ in the definition of $S_n$ and $D_n$ by $(d+1)*n$ and define a spire to be $v_n$ followed by $d$ other vertices between $D_n$ and $S_{n+1}$. These spires will occur with constant, non-zero probability, so the same computation as above proves $\Pr(E)=0$, i.e. the random basic walk cycles with probability 1. Summarizing:

\begin{thm} \label{thmregular} On any locally finite, $d-$regular graph $G$, the random basic walk cycles with probability 1.
\end{thm}

An interesting corollary is that a random basic walk on the bi-infinite binary tree must cycle. In the binary tree the probability of going up at any given vertex is $1/3$ while the probability of going down is $2/3$. It is tempting to compare the random basic walk on this vertex with a biased random walk on $\Z$ and conclude that the robot escapes to the downwards infinity. However, the spires argument demonstrates that this intuition is off the mark, since a walk on the binary tree must cycle.

As in the case of Theorem \ref{lattices}, this proof method gives an upper bound on the expected number of vertices visited, though not necessarily the expected maximum number (for a longest tour). If one considers more spires, e.g. spires which are jagged in a different way than the one shown in Example \ref{graphtraps}, then the probability of trapping via a spire is greater than $c$, so the upper bound can be improved. Furthermore, using the traps $V$ from Example \ref{graphtraps} will also give a proof and will vastly reduce the upper bound on the expected number of vertices visited because the shells will only need one layer of vertices between them rather than the 3 layers required for spires.

We conclude this section with a result that shows the question of expected maximum length of a cycle is unbounded. This surprising fact gives the opposite intuition about the random basic walk on regular graphs than the theorem above, but does not contradict the theorem because these cycles are very rare.

\begin{prop} \label{prop:pathlength}
Let $G$ be a locally finite, $d$-regular graph with port orientations selected uniformly at random as in Section \ref{introduction}.\ref{sec:basic}. Then the following statement holds with probability 1: for all $n\in \N$ there are infinitely many pairs $(v,\ell)$ such that a random basic walk starting at $v$ with initial port $\ell$ will visit $n$ vertices.
\end{prop}

\begin{proof} Fix $n$ and let $v$ be an arbitrary vertex. Because $G$ is infinite, there must be some path in $G$ of length $n$ which always moves away from $v$ in the graph metric. Let $P_{n,v}$ be the event that one such path has the appropriate labels so that a random basic walk starting at $v$ with initial port $\ell$ will move directly away from $v$ along the path for $n$ steps. The probability of $P_{n,v}$ is a constant, non-zero number $c_n=(1/d)^n$ which does not depend on $v$. Select a sequence of vertices $(v_1,v_2,\dots)$ which are all distance at least $2n$ away from each other, so the events $P_{n,v_i}$ are independent from each other. Then the probability that none of the events $P_{n,v_i}$ occur is $(1-c_n)*(1-c_n)*(1-c_n)*\dots = 0$. This proves there is some $v_j$ which has a path of length $n$ going out.

Let $w_1=v_j$. Removing $v_j$ from the list $(v_1,v_2,\dots)$, one can repeat the same argument and conclude that some other $v_k$ must have a path of length $n$ going out, since the probability of not having such a $v_k$ is $(1-c_n)*(1-c_n)*\dots=0$. Set $w_2 = v_k$. Repeating this ad infinitum proves that there is an infinite sequence $(w_1,w_2,\dots)$ each of which has a path of length $n$ going out.
\end{proof}

Note that the random basic walk starting at a random vertex $v_0$ is unlikely to visit any of the $w_i$. Extracting bounds from Theorem \ref{thmregular} on the expected length of the random basic walk shows that the random basic walk is unlikely to visit a large number of the vertices $v_i$. The probability that a given $v_i$ has a path of length $n$ going out is $(1/d)^n$, and the random basic walk will most likely not visit anywhere near $d^n$ many of the $v_i$. From the proposition above that there are paths of arbitrary length, it is easy to conclude that there must also be cycles of arbitrary length.

\begin{prop} \label{prop:cyclelength}
Let $G$ be a locally finite, $d$-regular graph with $d>2$ and with port orientations selected uniformly at random as in Section \ref{introduction}.\ref{sec:basic}. Then the following statement holds with probability 1: for all $n\in \N$ there is some vertex $v$ which is contained in a cycle of length greater than $n$. Indeed, there are infinitely many such $v$.
\end{prop}

\begin{proof}
Fix $n\in \N$ and let $M$ denote the first multiple of $d$ which is greater than $n$. By Proposition \ref{prop:pathlength} there is an infinite sequence of vertices $(w_1,w_2,\dots)$ each having a path of length $M$ going out. For each $w_i$ there is a constant, nonzero probability $k_n$ that this path is actually a spire, i.e. that the random basic walk will move out to the end of this path, then turn around and return to $w_i$. Because $d$ divides the length of the path, this spire will be a cycle using $M$ vertices, i.e. the random basic walk will traverse it back and forth forever. Note that $d>2$ is needed in order for the random basic walk to return along the spire rather than getting trapped in a cycle on the final two vertices of the spire.

The probability that none of the paths from Proposition \ref{prop:pathlength} are spires is $(1-k_n)*(1-k_n)*\dots = 0$. Thus, there must be some $w_i$ which is the first vertex in a cycle of length greater than $n$. Set $z_1=w_i$ and remove $w_i$ from the sequence $(w_1,w_2,\dots)$. Repeating the argument above proves that there is some $z_2 = w_j\neq w_i$ which is the first vertex in a cycle of length greater than $n$. Continuing forever gives an infinite sequence $(z_1,z_2,\dots)$ where each $z_i$ is the first vertex in a cycle of length greater than $n$.
\end{proof}

These propositions prove that there is no bound on the maximum length of a path or cycle which the random basic walk could hit. They do not give the expected average length of a cycle or even a bound for it. We return to this questions in Chapter \ref{conclusion}.

\section{Transience and Cycling on Locally Finite Graphs}
\label{sec:locallyfinite}


We conclude this chapter with a theorem which subsumes the previous theorems and for which the hypotheses cannot be weakened further. In particular, we shall no longer assume $G$ is regular, though we retain our standing assumptions that $G$ is connected, countable, and locally finite. Note that the spires argument does not generalize to non-regular graphs because there is no guarantee that such a spire at a vertex $v$ can be used to make a trap. For instance, if $v$ has degree at least 4 but the only neighbors of $v$ have degree 2 then there is no assignment of port numbers to the vertices on the spire which will force a cycle, because it is impossible to enter $v$ by an arc labeled 3 or 4. To make the spires argument generalize would require an assumption of a strong expander property on $G$.

Rather than placing such an assumption, we seek different types of trapping configurations. The solution is to generalize configuration $V$ given in Example \ref{graphtraps}. For a fixed vertex $v$ of degree $d(v)$ which is visited on the random basic walk by a port $i$ from $v_{-1}\to v$, consider the following configuration $C_v$: any port orientation at $v$ is acceptable as long as the arc from $v\to v_{-1}$ is labeled by $i-1$. Furthermore, for every neighbor $w\neq v_{-1}$, if the arc from $v$ to $w$ is labeled by $j$, then the arc from $w$ to $v$ must be labeled by $j' = (j \mod{d(w)})+1$.

The configuration $C_v$ forces a trap at $v$ because it is impossible for a random basic walk to get distance 2 away from $v$ once reaching $v$. However, this configuration may be impossible for some $v$. Several things can go wrong, all having to do with independence of the occurrence of $C_v$ from the random basic walk leading up to $v$. One problem is that a neighbor $w$ may have been visited before and the port number $j'$ out of $w$ may have already been assigned. Indeed, the probability that the correct port number is on the arc $w\to v$ is only $1/d(w)$.

Another problem is that $v_{-1}$ might have degree larger than $d(v)+1$. If this occurs, and if $i>d(v)+1$ then the arc from $v$ to $v_{-1}$ cannot have the label $i-1$ required, because $i-1 > d(v)$. If $d(v_{-1})\leq d(v)$ then $C_v$ occurs with probability $p_v = (1/d(v))*\prod{1/d(w)}$ where this product is over all neighbors $w\neq v_{-1}$. To see this, it is best to view the random basic walk as only labeling one arc per step, as discussed in Section 1.\ref{sec:basic}. In order to guarantee that the configurations $C_v$ will occur, the assumption needed on $G$ is that there is a constant $D$ such that $d(v)<D$ for all $v$ (i.e. $G$ has \emph{bounded degree}).




\begin{thm} \label{maintheorem} On any graph $G$ with all vertex degrees bounded by a constant $D$, the random basic walk cycles with probability 1.
\end{thm}

\begin{proof}[\textbf{Star Method}] The pigeonhole principle guarantees us that there are infinitely many vertices $v$ with a neighbor $w$ of degree $d(w)\geq d(v)$. This is because every time a vertex $v$ only has neighbors of smaller degree, all those neighbors have a neighbor (namely, $v$) with larger degree. If the random basic walk is to have any chance of escaping to infinity, then it must be the case that infinitely often the robot moves from a vertex $v$ to a vertex $w$ such that $d(v)\geq d(w)$. This uses the hypothesis of bounded degree, since the robot cannot move along a chain from a vertex of degree 1, to one of degree 2, to one of degree 3, etc. Such a chain would eventually hit a vertex of degree $D$ and then need to move to one of degree $\leq D$.

Label the steps of the robot which move from larger degree to smaller degree by $w_1\to v_1,w_2\to v_2, \dots$, where $w_i= w_j$ is allowed but $v_i\neq v_j$ for $i\neq j$ is disallowed by removing the pairs $(w_i,v_i)$ where $v_i$ has appeared in the list before. Clearly removing pairs with repeated second coordinate will not change the fact that there are infinitely many such pairs.

Whenever the random basic walk takes a step $w_i\to v_i$ from the list above, let $E_i$ be the event that the configuration $C_{v_i}$ is achieved. Let $N$ denote the set of neighbors of $v_i$ which are not $w_i$. Because $d(v_i)<D$, it is clear that $|N|<D$. Furthermore, $d(x)<D$ for all $x\in N$, so

$$\Pr(E_i) = \frac{1}{d(v)}*\prod_{x\in N(v)}\frac{1}{d(x)} > \frac{1}{D}*\left(\frac{1}{D}\right)^D$$

Denote the number on the right by $c$, and note that this number is strictly greater than 0. In order for the random basic walk to escape to infinity, the event $E_i$ must be avoided for infinitely many $v_i$. So the probability $p$ that the random basic walk escapes to infinity satisfies $p\leq (1-c)*(1-c)*\dots =0$, proving that the random basic walk cycles with probability 1.

\end{proof}


It is worth noting that achieving $C_v$ may be more than is necessary, since not every neighbor of $v$ needs to be used in a cycle (see Example \ref{stargraph}). For this reason, an upper bound on the expected number of vertices visited created from the proof method above could be quite bad. Some graphs do require all the neighbors of $v$ to be involved in a trap, as discussed further in Chapter \ref{sec:complete}.

The final question we consider in this chapter is whether or not the hypotheses of the theorem can be weakened. The hypothesis of local finiteness is necessary in order to define the labeling, since arcs out of vertices of infinite degree could not be given labels in the same way. Furthermore, the following example proves that the hypothesis of bounded degree cannot be dropped. This graph has unbounded degree and the probability of a path to infinity is strictly greater than zero.

\begin{thm} \label{unboundeddegree} Let $T$ be the tree where every vertex in level $n$ has $2^n$ children. As usual, the root is in level 0.

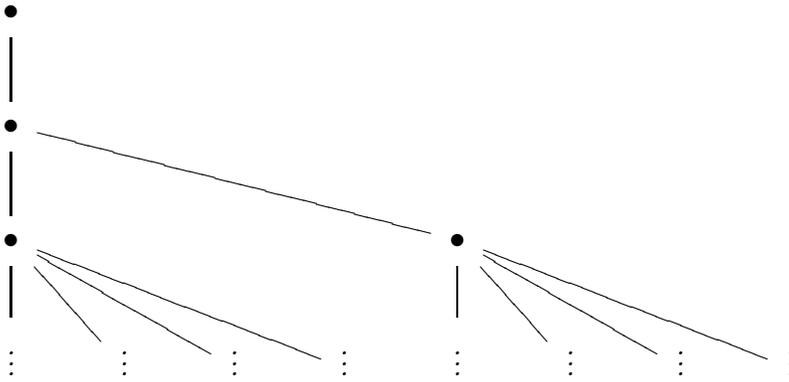
\xymatrix{
\bullet \ar@{-}[d] &&&&&&&\\
\bullet \ar@{-}[d] \ar@{-}[drrrr]&&&&&&&\\
\bullet \ar@{-}[d] \ar@{-}[dr]\ar@{-}[drr]\ar@{-}[drrr]&&&&\bullet\ar@{-}[d]\ar@{-}[dr]\ar@{-}[drr]\ar@{-}[drrr]&&&\\
\vdots&\vdots&\vdots&\vdots&\vdots&\vdots&\vdots&\vdots\\
&&&&&&&}

Then the random basic walk on $T$ is transient, i.e. has nonzero probability of escaping to $\infty$.

\end{thm}

\begin{proof}
A random basic walk on $T$ will always have a higher probability of going downwards than of going upwards. It is clear that the random basic walk starting at the root will be more likely to cycle than a random basic walk starting at a vertex lower down on $T$, since the probability of returning will always be higher near the root. Thus, we will focus on the case where the initial vertex is the root. From the root, the initial step is determined because the root has degree 1.

Define an event $P$ to be ``for each $i$, the robot is distance $i$ away from the root at step $i$'' i.e. ``all steps are away from the root.'' Note that $\Pr($escape$)\geq \Pr(P)$ since $P$ is a way for the robot to escape. Define events $P_i$ to be ``at step $i$ the robot is distance $i$ from the root.'' Clearly, $P_0\supset P_1\supset P_2\supset \dots \supset P = \bigcap P_i$, so Pr$(P)=\prod$Pr$(P_n \;|\; P_{n-1})$. Because each vertex has only one arc pointing back at the root, $\Pr(P_0) = Pr(P_1)=1, \Pr(P_2 \;|\; P_1) = 1-\frac{1}{3} \geq 1-\frac{1}{2}, \Pr(P_3 \;|\; P_2) = 1-\frac{1}{5} \geq 1-\frac{1}{4}$, and in general

$$\Pr(P_m \;|\; P_{m-1}) = 1-\frac{1}{2^{m-1}+1}\geq 1-\frac{1}{2^{m-1}} \mbox{. Thus: } \Pr(P) \geq \prod_{m=2}^\infty{1-\frac{1}{2^{m-1}}}$$

$$ = \prod_{n=1}^\infty{1-\frac{1}{2^n}} \mbox{ and taking logs yields } \ln(\Pr(P))\geq \sum_{n=1}^\infty{\ln\left(1-\frac{1}{2^n}\right)}$$

Note that the inequality is preserved after applying $\ln$ because $\ln$ is an increasing function. Proving $\Pr(P)>0$ is equivalent to proving this sum is greater than $-\infty$. Recall the Taylor Series expansion of $\ln(1-x)$ around $0$:

$$\ln(1-x) = -x - \frac{x^2}{2} - \frac{x^3}{3} - \dots \mbox{ for } -1\leq x<1$$

Because $-1\leq 1/2^n\leq 1$ for $n\geq 1$, this equality holds. Thus:

$$\ln(\Pr(P)) \geq -\sum_{n=1}^\infty{\frac{1}{2^n}} -\frac{1}{2}\sum_{n=1}^\infty{\frac{1}{2^{2n}}} -\frac{1}{3}\sum_{n=1}^\infty{\frac{1}{2^{3n}}} -\dots-\frac{1}{k}\sum_{n=1}^\infty{\frac{1}{2^{kn}}} -\dots \geq $$

$$\sum_{n=1}^\infty{\frac{-1}{2^n}}+\sum_{n=1}^\infty{\frac{-1}{2^{2n}}}+\sum_{n=1}^\infty{\frac{-1}{2^{3n}}}\dots = \frac{-1/2}{1-1/2}+\frac{-1/2^2}{1-1/2^2}+\frac{-1/2^3}{1-1/2^3}\dots \geq \sum_{n=0}^\infty\frac{-1}{2^n} = -2$$


To show $2^{-k}/(1-2^{-k}) \leq 2^{-k+1}$ as needed in the last inequality, note that $1-2^{-k} \geq 1/2$ so that $2^{-k}/(1-2^{-k}) \leq 2^{-k}/(2^{-1}) = 2^{-k+1}$.

Undoing the log shows that Pr(escape)$\geq \Pr(P) \geq e^{\ln(\Pr(P))} \geq e^{-2}>0$. This proves there is a positive chance that the robot escapes, i.e. the random basic walk is transient.
\end{proof}

\chapter{Random Basic Walks on Finite Graphs}
\label{sec:complete}

In this chapter we return to the question of finite graphs which originally motivated the random basic walk in \cite{supersize}. In that paper, the authors hoped that in a random basic walk the robot would do a good job exploring the finite graph $G_{k,n}$, i.e. would cover a constant fraction of the nodes. In light of the theorems in Chapter \ref{ch:infinite} this seems unlikely, as the robot does a very poor job of exploring infinite graphs. Still, in \cite{supersize}, experimental evidence is given which suggests that the random basic walk on $[n]\times [n]$ visits $1.2701 \cdot |G|^{1.8891}$ of the nodes.

This experimental evidence does not appear to be the same as $c*|G|$ so in this chapter we prove that $\mathcal{K} = \{K_n\}$ is an infinite class of graphs where the random basic walk is asymptotically expected to explore a constant fraction of the nodes. Furthermore, $\mathcal{K}$ is useful in practice because it is a class of graphs where the robot is given more freedom of movement than on square lattices, i.e. it may more closely mimic some real world applications.

\begin{thm} \label{completevertices} As $n\to \infty$, a random basic walk on $K_n$ is expected to visit at least $(1-1/e)*n$ nodes.
\end{thm}

\begin{proof} We prove a stronger statement, namely that $(1-1/e)*n$ of the nodes are expected to be visited within just the first $n-1$ steps. Restricting to the first $n-1$ steps guarantees that cycles are impossible, since any cycle in a $d$-regular graph a cycle requires at least $d$ arcs. Label the vertices $0,1,2,\dots,n-1$ and assume that the starting vertex $v_0$ is labeled by $0$. Consider the sequence of vertex numbers $(v_0,v_1,v_2,\dots)$ visited by the robot. This sequence cannot have adjacent numbers equal, since $K_n$ has no loops. Furthermore, this sequence cannot have a pair of adjacent numbers occur twice, since acyclicity implies no arc is traversed twice.

Let $v\neq v_0$ be a vertex chosen at random. We will prove that the probability that $v$ is visited is $\geq 1-1/e$. The probability that $v$ is visited on step 1 of the robot is $\Pr(v = v_1)=1/(n-1)$, since there are $n-1$ possible steps the robot could take after $v_0$. If $v$ is not visited in step 1, then the probability that $v$ is visited on step 2 of the robot is $\Pr(v=v_2\;|\; v\neq v_1)=1/(n-1)$. Step 3 is more complicated, because it is possible that $v_2=v_0$, in which case the arc $v_2\to v_1$ already has a label. So there are two ways to get $v=v_3$ given that $v$ has not been visited previously:

$\Pr(v=v_3) = \Pr(v=v_3\;|\;v_2=v_0)*\Pr(v_2=v_0)+\Pr(v=v_3\;|\;v_2\neq v_0)*\Pr(v_2\neq v_0)$
$$= \frac{1}{n-1}*\frac{1}{n-2} + \frac{1}{n-1}*\frac{n-2}{n-1} > \frac{1}{n-1}*\frac{1}{n-1}+\frac{n-2}{n-1}*\frac{1}{n-1} = \frac{1}{n-1}$$

The cases for $v_4,v_5,\dots$ get even more complicated, but the point remains that $1/(n-1)$ is always a lower bound on the probability that $v$ is first visited in step $i+1$. This is because if $v_i$ has been visited $k$ times before, then the existence of these previous visits rules out more of the possible arcs the robot could follow. This can only increase the probability that the next vertex to be visited is one the robot has not visited before (e.g. $v$), and this makes the messy computations summations which factor in the entire path of the robot unnecessary. Formally, the probability that $v$ is first visited on the $i+1$-st step given that $v_i$ has been visited $k$ times previously will be $1/(n-k) \geq 1/(n-1)$. This proves that $\Pr(v$ not visited on $i$-th step$) \leq 1-1/(n-1)$ for all $i$, which implies

$$\Pr(v \mbox{ not visited in first $n-1$ steps}) = \prod_{i=1}^{n-1} \Pr(v_i\neq v) \leq \left(1-\frac{1}{n-1}\right)^{n-1}$$

As $n\to \infty$, this bound tends to $1/e$, so the probability that $v$ is visited in the first $n-1$ steps is at least $1-1/e$. Thus, the expected number of vertices missed is $\leq n/e$ and the expected number of vertices visited is $\geq (1-1/e)*n$.
\end{proof}

Note that the method of proof above mimics the proof that the cover time for a simple random walk on $K_n$ is $O(n\log n)$ (see e.g. \cite{lovasz}). The theorems in Chapter \ref{ch:infinite} are evidence that the random basic walk is more constrained than the simple random walk. This fact also means that the random basic walk is better than a simple random walk when it comes to exploring $K_n$ in the first $n-1$ steps. However, the random basic walk can get stuck in a cycle and fail to visit every vertex whereas the simple random walk is guaranteed to visit every vertex with probability 1, so the improvement of the random basic walk over the simple random walk comes with a price. The following example shows that the random basic walk can fail to visit every vertex, even on $K_n$:

\begin{example} The following is a $K_{10}$, where only 10 edges are shown and these 10 edges form a cycle. The basic walk only visits half the vertices.

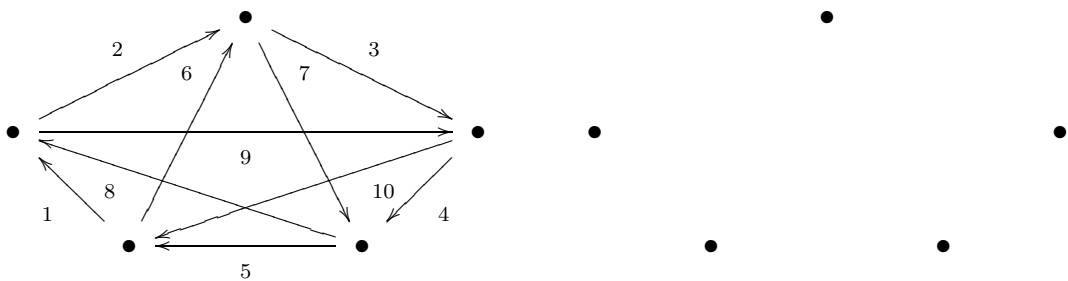
\xymatrix{& & \bullet \ar[drr]^3 \ar[ddr]^(.3)7& &                         & & & \bullet & & \\
\bullet \ar[urr]^2 \ar[rrrr]_9& & & & \bullet \ar[dl]^4 \ar[dlll]^(.3){10} & \bullet & & & & \bullet\\
& \bullet \ar[ul]^1 \ar[uur]^(.7)6& & \ar[ll]^5 \ar[ulll]^(.7)8\bullet &   & & \bullet & & \bullet & }
\end{example}

Having considered how the random basic walk on $K_n$ visits vertices, we shift our attention to how it samples from the set of labeled arcs. A cycle occurs exactly when an arc is traversed in the same direction for the second time, and we would like to know how many arcs we can expect the robot will traverse before this occurs. Experimental evidence privately communicated by Sunil Shende suggests the following conjecture:

\begin{conjecture} The expected number of arcs traversed by a random basic walk on $K_n$ is $1.8*n$ as $n\to \infty$.
\end{conjecture}

Before remarking on how one might prove this conjecture, we discuss a different way to prove Theorem \ref{completevertices}. First, consider the following process $X$ to create a sequence $(x_n)$:

\begin{itemize}
  \item $x_1=1$
  \item Choose $x_2$ uniformly at random from $\{2,\dots,n\} = [n]-\{x_1\}$
  \item Choose $x_3$ uniformly at random from $[n]-\{x_2\}$
  \item Choose $x_4$ uniformly at random from $[n]-\{x_3\}$
  \item Continue in this way, stopping at $x_n$
\end{itemize}

This process is very related to the process which creates the sequence $(v_n)$, i.e. the sequence obtained by following the random basic walk and writing down labels for vertices visited. This is because when the robot first reaches $x_k$ (entering by label $i_k$), there are $n-1$ neighbors and all are equally likely to have arc $x_k\to x_{k+1}$ be labeled by $(i_k \mod{n})+1$. Let $V_i$ be the number of vertices visited after the $i$-th step, i.e. the number of unique integers appearing in $(v_1,\dots,v_i)$. Let $X_i$ be the number of unique integers appearing in $(x_1,\dots,x_i)$.

The Occupancy Problem states that the expected value of $X_n$ is $(1-1/e)*n$. The difference between the random basic walk process and $X$ is that if the random basic walk has visited $x_k$ before, then the probabilities are different because now at least one of the arcs leaving $x_k$ cannot receive the new label. Because these probabilities are not very different, $\Pr(V_k = c)$ and $\Pr(X_k = c)$ only differ by a small amount. It is possible that a coupling argument can be used to show that the expected value of $V_n = (1-1/e)*n$. Basic information about coupling can be found in Chapters 4 and 5 of \cite{mixing}. The first step if one wished to pursue this method would be to compute the total variation distance between the two distributions.


We now return to the conjecture, which concerns the sequence $((v_1,\ell_1),(v_2,\ell_2),\dots)$ of vertices and ports. The data of a pair $(v_i,\ell_i)$ is exactly the data of a labeled arc. One way to find the first instance of a repeated $(v_i,\ell_i)$ pair, is via a coupling argument with the following process $Z$ from elementary probability theory:

\begin{itemize}
  \item Select an integer $r_1$ uniformly at random from $[n]$ and set $z_1=(1,r_1)$
  \item Select an integer $r_2$ uniformly at random from $[n]$ and set $z_2=(2,r_2)$
  \item Continue in this way all the way up to $z_n = (n,r_n)$
  \item Select an integer $r_{n+1}$ uniformly at random from $[n]$ and set $z_{n+1}=(1,r_{n+1})$
  \item Continue in this way up to $z_{2n} = (n,r_{2n})$
  \item Keep going until some $z_k$ has already appeared in the list, then stop.
\end{itemize}

This process $Z$ is similar to the Birthday Problem if it were being run in $n$ different rooms with stopping criterion given by a match in any room. In the first $n$ steps there are no chances for repetition, because of the first coordinate in each $z_i$. Between step $n+1$ and step $2n$, every step brings a probability of $1/n$ of a repetition. After step $2n$ (if the walk gets that far), each step brings a probability of $2/n$ of a repetition. The distribution $Z$ is something like a sequence of geometric distributions where the probabilities of repetition change every $n$ steps. A recursion discovered by Danny Krizanc and Sunil Shende proves that the expected $k$ such that $Z$ terminates is exactly the number $1.8*n$ suggested by the experiments on $K_n$ (up to 10 digits). To formalize the connection between the process $Z$ and the process given by the random basic walk, a coupling argument should be used.

\chapter{Conclusion}
\label{conclusion}

We have introduced the notion of a random basic walk of \cite{supersize} to infinite graphs $G$, have introduced an equivalent formulation of how the labeling is done in the random basic walk so that it is allowed to set labels one at a time rather than all at once, have catalogued potential applications of the random basic walk, and have provided detailed comparisons between the random basic walk and existing generalizations of simple random walks. Furthermore, we have introduced the notion of a random rotor router, and have discussed notions of transience for rotor routers and self-avoiding random walks which do not seem to appear in the literature elsewhere.

We have defined analogues in the setting of random basic walks of the recurrence and transience properties of simple random walks, have proven a theorem which states that any graph of bounded degree has a cycling random basic walk, and have shown that these hypotheses cannot be removed. We have studied numerous examples of the type of behavior which can occur, and have demonstrated that the cycling and transience results only hold with probability 1. From the theorems above, we have extracted upper bounds on the expected number of vertices a random basic walk will visit on $G$, and these bounds will apply to the graphs $G_{k,n}$ of \cite{supersize} for $k$ and $n$ sufficiently large.

We have also extended the knowledge about random basic walks on finite graphs, in particular proving that $\{K_n\;|\; n\in \N\}$ is an infinite class of graphs on which random basic walks asymptotically visit a constant fraction of the nodes. We have stated a conjecture based on experimental evidence which regards the asymptotic expected number of arcs traversed in $\{K_n\;|\; n\in \N\}$, and have sketched how a proof of this conjecture might proceed. We end now by stating several problems regarding the random basic walk on finite graphs which are still open, and discussing why these questions are interesting.

\begin{question}
In $\Z^2$, what is the expected length of time (i.e. the expected number of steps taken) before the random basic walk hits a cycle? What is the expected size of a cycle? What about these questions on $G_{k,n}$? What about the expected maximum length of a random basic walk on $G_{k,n}$?
\end{question}

If the first two questions could be answered, then we would also know the expected number of vertices visited, and potentially the expected number of arcs used. An interesting invariant to study for the random basic walk might be the expected number of arcs used per vertex visited. These same questions are of course also interesting on $G_{k,n}$, and the answers should be related as $k,n\to \infty$. Indeed, the values for $G_{k,n}$ are already bounded from above by the proof method of Theorem \ref{2d}.

Proposition \ref{prop:pathlength} and Proposition \ref{prop:cyclelength} show that results on $\Z^2$ cannot be used to answer the final question above or even to give a nontrivial bound. However, for fixed $k$, the theory of percolation may be useful as a way of proving $n$ is a lower bound on the expected maximum distance the random basic walk will travel on $G_{k,n}$. The theory of percolation discusses when a particle will be able to move from the left-most column to the right-most column, which requires at least $n$ steps. Currently, no nontrivial bounds are known.

One way to proceed on the first two problems is to find the probability that a given vertex $v$ (with starting port $i$) is part of a cycle, i.e. a random basic walk out of $v$ will return to $v$ along an arc labeled $i-1$. Due to the symmetry in $\Z^2$ this number $p$ will be the same for every $v$. One way to compute $p$ is to look at all possible walks using just three vertices and cycling, then all using four vertices and cycling, etc. The author hopes that there will be a pattern after enough small examples are worked out, and suspects that Young tableaux will be of use.





If one can compute the probabilities $p_i$ of traps using $i$ vertices, then the expected length of a cycle is $\sum{k*p_k}$ as $k$ runs through $\N$. Furthermore, $p = \sum p_i$ and knowledge of $p$ can be used to find the expected length of a walk before the random basic walk falls into a cycle, i.e. before a vertex $v$ is reached via a label which causes a cycle at $v$. The step of the walk before reaching $v$ is a step $w\to v$ such that $w$ will never be visited again, but $v$ will be visited infinitely many times. The goal of understanding cycles better leads to the following two questions:

\begin{question}
How does punching large holes in the graph $\Z^2$ affect the expectations above?
\end{question}

This question is of interest because it relates to the initial motivation for considering the random basic walk, i.e. a robot vacuuming a room. The large holes will represent furniture which the robot must move around. Furthermore, it is expected that having holes in the graph might lead to long cycles which move around those holes.

\begin{problem}
Study the random basic walk with the constraint that there are no vertices which have two different incoming arcs with the same label.
\end{problem}

This problem is interesting because the basic walk falls into a cycle at $v$ if and only if $v$ has the property disallowed above. If such labelings are disallowed then the basic walk can never fall into a cycle, but cycles are still possible if the initial vertex $v_0$ with the initial label is part of a cycle. The constrained random basic walk has never been studied, and questions of cycling and transience are completely open.

\begin{problem}
Determine whether or not there is a graph $G$, analogous to the Rado graph, such that every finite graph with local orientations sits inside $G$ in the same way that finite graphs sit inside the Rado graph, i.e. such that any automorphism of a finite graph with labeling extends to an automorphism of $G$. Determine whether or not $G$ is unique up to isomorphism.
\end{problem}

In order for this problem to be solved, a good notion of graph homomorphism for graphs with local orientations would need to be developed. The Rado graph is constructed as a Fra\"{i}ss\'{e} Limit, and this construction might also work to construct $G$, but the details would need to be checked.

The next three problems are motivated by the theory of simple random walks on finite graphs, as discussed in Section \ref{introduction}.\ref{sec:randomwalks}.

\begin{problem} Create and study an analogue of the hitting time for the random basic walk.
\end{problem}

The author suggests the following: either the probability $p_1^v$ of reaching a given vertex $v$ before returning to the starting vertex or the probability $p_2^v$ of reaching $v$ at all. For simple random walks, computation of hitting time it is often related to the Dirichlet problem for harmonic functions, and this relationship is used to give a bound on the hitting time in \cite{whipple}. Related to the problem above are the following two problems:

\begin{problem} Create and study an analogue of the mixing time for the random basic walk.
\end{problem}

\begin{problem} Create and study an analogue of load balancing for the random basic walk.
\end{problem}

The second problem seems more tractable and should be related to the number of vertices a random basic walk is expected to visit.

Several questions were mentioned in the main body of the thesis, and we repeat them here for completeness.

\begin{question} Is there a labeling on $\Z^2$ such that the basic walk from any initial $v_0$ and any initial port will visit every vertex? Is there an infinite family of such labelings? What about for $\Z^d$ with $d>2$?
\end{question}

This question was originally mentioned in Section \ref{sec:examples}.\ref{integerlabelings}. Another question which arose in Section \ref{ch:infinite}.\ref{sec:regular} was to find analogues of the labelings given for $\Z^d$ which work on the hexagonal lattice:

\begin{question} On the hexagonal lattice, are there labelings which allow the basic walk to escape to infinity from any starting vertex and with any initial port? Are there labelings such that the basic walk visits every vertex regardless of where it starts? Is there a method of constructing such examples for general graphs $G$?
\end{question}

It is likely the desired labelings can be found for the hexagonal lattice, especially the first one asked for above. It seems unlikely there is any construction which will hold for general graphs $G$. The following is the last question from the main body of the thesis, which appeared in Section \ref{ch:history}.\ref{sec:quasirandom}:

\begin{problem} Study random rotor routers and determine if they bear any resemblance to the random basic walk.
\end{problem}

We hope that at least some of these questions and problems will be pursued and will yield interesting answers.

\backmatter



\end{document}